\documentclass[10pt]{article}
\usepackage{my-common}

\title{Spatio-Temporal Change of Support Modeling with R}
\author{Andrew~M. Raim$^{a}$, Scott~H. Holan$^{b,c}$, Jonathan~R. Bradley$^d$, \& Christopher~K. Wikle$^b$}

\begin{document}
\maketitle

\begin{abstract}
Spatio-temporal change of support methods are designed for statistical analysis on spatial and temporal domains which can differ from those of the observed data. Previous work introduced a parsimonious class of Bayesian hierarchical spatio-temporal models, which we refer to as STCOS, for the case of Gaussian outcomes. Application of STCOS methodology from this literature requires a level of proficiency with spatio-temporal methods and statistical computing which may be a hurdle for potential users. The present work seeks to bridge this gap by guiding readers through STCOS computations. We focus on the \proglang{R} computing environment because of its popularity, free availability, and high quality contributed packages. The \pkg{stcos} package is introduced to facilitate computations for the STCOS model. A motivating application is the American Community Survey (ACS), an ongoing survey administered by the U.S. Census Bureau that measures key socioeconomic and demographic variables for various populations in the United States. The STCOS methodology offers a principled approach to compute model-based estimates and associated measures of uncertainty for ACS variables on customized geographies and/or time periods. We present a detailed case study with ACS data as a guide for change of support analysis in \proglang{R}, and as a foundation which can be customized to other applications.
\end{abstract}


\blfootnote{
\begin{flushleft}
$^a$Center for Statistical Research and Methodology, U.S. Census Bureau \\
$^b$Department of Statistics, University of Missouri \\
$^c$Office of the Associate Director for Research and Methodology \\
$^d$Department of Statistics, Florida State University \\
$^*$Emails: \url{andrew.raim@census.gov}, \url{holans@missouri.edu}, \url{bradley@stat.fsu.edu}, \url{wiklec@missouri.edu}
\end{flushleft}
}

\section{Introduction}
\label{sec:intro}
In the course of an analysis where data are inherently spatio-temporal, an investigator may desire estimates on spatial and/or temporal domains not coinciding exactly with domains of the observations. This can include customized geographies and time periods conceived long after the data have been collected. Spatio-temporal change of support methods aim to provide this capability. A methodology recently proposed by \citet{BradleyWikleHolan2015Stat} captures spatio-temporal dependencies in areal data by constructing several key matrices which become the foundation of a Bayesian hierarchical model. Model fitting is done via Markov chain Monte Carlo (MCMC); in particular, the model permits a Gibbs sampler which is conveniently composed of draws from standard distributions. Estimates, predictions, and appropriate measures of uncertainty are provided by the fitted model. This methodology, hereafter referred to as the \textit{STCOS model} or \textit{STCOS methodology}, is the focus of the present paper. Although STCOS methodology has been fully specified by \citet{BradleyWikleHolan2015Stat}, potential users---such as subject-domain scientists who may not be experts in spatio-temporal statistics---may find proceeding from the previous literature to their own applications to be a difficult hurdle. A successful implementation requires managing datasets containing estimates, geospatial data, operations on sparse matrices, Bayesian computing, plotting, as well as carrying out computations tailored to the STCOS model.

In this paper, we demonstrate an assortment of tools to perform STCOS modeling through a detailed case study, with the objective of making the methodology more accessible to potential users. The required tasks can be accomplished with a variety of modern computing platforms, but we will focus on \proglang{R}, the popular open source environment for statistical computing \citep{Rcore2020}. \proglang{R} is supported by an active community of academic, corporate, and individual users. A large and diverse collection of packages has been contributed by its community and published to repositories such as the Comprehensive \proglang{R} Archive Network (CRAN). Much of \proglang{R}, including the base platform and CRAN packages, is freely available on the internet. The high-level \proglang{R} programming language facilitates data analysis, fast prototyping of new methods, and simulation, and can be augmented with \proglang{C}, \proglang{C++}, and \proglang{FORTRAN} when speed or efficient use of memory are crucial. In addition to highlighting some established \code{R} packages, we introduce the \pkg{stcos} package to handle some of the more intricate STCOS computations in an efficient and user-friendly way. Familiarity with \proglang{R} will be assumed throughout the remainder of the paper.

\citet{BradleyWikleHolan2015Stat} developed STCOS with a motivating application to the American Community Survey (ACS), an ongoing survey administered by the U.S.~Census Bureau for the purpose of measuring key socioeconomic and demographic variables for the U.S.~population. ACS is again showcased in the present paper as it remains an important application for change of support methods~\citep{NCRNPub2019}. However, STCOS methodology is not limited to applications involving the ACS or the U.S. Census Bureau. The problem of spatial change of support has arisen in atmospheric science and oceanography \citep{WikleBerliner2005}, water quality modeling \citep{RodeEtAl2010}, environmental health \citep{FuentesEtAl2006}, and remote sensing \citep{NguyenEtAl2012}, among others. See \citet{GotwayYoung2002}, \citet{BradleyWikleHolan2015Stat}, and the references therein for a review of the change of support literature.

Existing software tools for change of support appear to originate from geographic information systems (GIS) literature, emphasizing methods such as pycnophylactic interpolation \citep{Tobler1979}, areal-weighted interpolation \citep{Lam1983}, and dasymetric mapping \citep[e.g.][]{EicherBrewer2001}. In \proglang{R}, such tools include the \pkg{pycno} package \citep{Brunsdon2014}, the \code{st_interpolate_aw} function in the \pkg{sf} package \citep{Pebesma2018}, and the \pkg{areal} package \citep{PrenerRevord2019}. The \pkg{Tobler} package was developed for \proglang{Python} by \citet{CortesEtAl2019}. \citet{QiuEtAl2012} and \citet{MileuQueiros2018} describe adding change of support capabilities to the \proglang{ArcGIS} and \proglang{QGIS} platforms, respectively. From the perspective of software tools, the present work offers two major contributions: (1) measures of uncertainty expressed via a statistical model, and (2) the capability to carry out change of support in both space and time.

The remainder of the article proceeds as follows. Section~\ref{sec:acs-stcos} discusses STCOS concepts in the context of the ACS. Section~\ref{sec:model} reviews STCOS methodology. Here some additional details are provided---and some small modifications are made---from the original formulation of \citet{BradleyWikleHolan2015Stat}. Section~\ref{sec:Rtools} discusses the set of \proglang{R} tools to be demonstrated, including basic functionality of the \pkg{stcos} package. Section~\ref{sec:columbia} presents our case study to demonstrate STCOS programming; we produce model-based estimates of median household income for several neighborhoods in the City of Columbia in Boone County, Missouri. Section~\ref{sec:conclusions} concludes the article. This article is intended to be largely self-contained for a wide range of readers; those eager to begin programming can focus primarily on Sections~\ref{sec:Rtools} and \ref{sec:columbia}. The \pkg{stcos} package is available on the CRAN at \url{https://CRAN.R-project.org/package=stcos}. The complete code for the City of Columbia data analysis is provided as a supplement to this article.

\section{Change of support concepts and the ACS}
\label{sec:acs-stcos}

To facilitate our discussion of the change of support problem and STCOS methodology, we now give a brief overview of the ACS. Public-use ACS data are available through the Census Bureau's ACS website (\url{https://www.census.gov/programs-surveys/acs}) dating back to the year 2005. Estimates have historically been released for 1-year, 3-year, or 5-year periods; 3-year period estimates were discontinued after 2013. The Census Bureau releases annual ACS period estimates for a variety of geographies including states, counties, census tracts, and school districts. At their finest geography, data are released at the census block-group level; however, estimates for an area are suppressed unless the area meets certain criteria. An area typically must have a population of at least 65,000 for 1-year estimates to be released, but there is no population requirement for 5-year estimates \citep{ACSSuppression}. ACS estimates consist of point estimates and associated measures of uncertainty such as margins of error (MOEs) corresponding to 90\% confidence intervals, or variance estimates; we will refer to them collectively as direct estimates. Because statistical agencies like the Census Bureau have direct access to the confidential microdata, special tabulations for new geographies or period lengths can be prepared internally as needed. However, data users outside of the Census Bureau may be interested in custom geographies and/or nonstandard time periods which are not provided by the agency. Providing ACS data users tools for change of support has recently been identified as an important problem by a National Academy of Sciences panel~\citep{NAS2015}. STCOS methodology enables model-based estimates to be computed with public-use ACS releases.

The change of support problem can be illustrated by a concrete example, taking median household income as the variable of interest here and for the remainder of the article. Suppose we would like to produce 3-year model-based estimates in Missouri congressional districts for the year 2015. Congressional districts are geographic regions which receive representation by an elected official in the U.S.~House of Representatives and are determined by a redistricting process which is based on data from each decennial census. The Census Bureau does release ACS estimates on congressional districts, but releases of 3-year estimates for all geographies were discontinued after 2013; therefore, model-based estimates may be of interest to data users. Geographies on which we want to produce estimates and predictions are referred to as \emph{target supports}. Figure~\ref{fig:county-vs-cd} displays the eight designated congressional districts in Missouri for the year 2015. Geographies on which direct estimates are available are used to fit the STCOS model and are referred to as \emph{source supports}. For this illustration, we could take the source supports to be all 1-year, 3-year, and 5-year ACS releases for the counties within Missouri. Including available periods over a number of years allows the STCOS model to find trends in both time and space, and make use of estimates which represent varying levels of granularity and sparseness. Figure~\ref{fig:acs-maps} shows direct estimates for Missouri in the year 2013. We notice that 1-year and 3-year period estimates have been suppressed for many counties. We emphasize that counties and congressional districts do not necessarily align, and the crux of the STCOS problem is to ``translate'' between the county-level observations and the congressional districts. The third type of support which must be discussed is the \emph{fine-level support}. For this example, we could take the fine-level support to be the 2015 definition of counties in Missouri, shown in Figure~\ref{fig:county-vs-cd}. The STCOS methodology works by translating each of the source supports to the fine-level support during the model fitting process. Once the model has been fit, estimates and predictions on target supports of interest are obtained by translating from the fine-level support. \citet{JSM2017-STCOS} presents a model selection study with counties in the continental U.S. as target and source supports, congressional districts as target supports, and median household income as the ACS variable of interest. Section~\ref{sec:columbia} will demonstrate a smaller-scale problem which requires less computing time.

\section{The STCOS model}
\label{sec:model}
Let $\mathcal{T} = \{ v_1, \ldots, v_T \}$ represent the set of times for which direct estimates are available, indexed by $t = 1, \ldots T$. Let $\mathcal{L}$ denote the set of possible lookback periods for which these estimates have been constructed. We will take $\mathcal{T}$ to consist of the years 2005 through 2017, corresponding to the ACS releases available during the preparation of this article, and $\mathcal{L} = \{1, 3, 5\}$ to denote 1-year, 3-year, and 5-year period releases. Therefore, $\ell$-year direct estimates for year $v_t$ are based on the time period $(v_{t-\ell+1}, \ldots, v_t)$. Data may not be released for all $(v_t,\ell) \in \mathcal{T} \times \mathcal{L}$; for example, ACS 3-year estimates were discontinued after 2013. Let $(\mathcal{T} \times \mathcal{L})^*$ denote the subset of $\mathcal{T} \times \mathcal{L}$ that corresponds to a data release. For each $(v_t,\ell) \in (\mathcal{T} \times \mathcal{L})^*$, the associated source support $D_{t\ell}$ is a collection of areal units whose estimates are included in the release. For each areal unit $A \in D_{t\ell}$, $Z_t^{(\ell)}(A)$ is the direct point estimate for one ACS variable of interest and $V_t^{(\ell)}(A)$ is the corresponding variance estimate. The fine level support will be denoted $D_B = \{ B_1, \ldots, B_{n_B} \}$. The total surface area of a given areal unit $A$ will be denoted $|A|$.

The STCOS model is a Bayesian hierarchical model \citep[][Section~2.1]{CressieWikle2011} which will first state before describing components in detail. Let $\text{N}(\vec{\mu}, \vec{\Sigma})$ denote the multivariate normal distribution with density 
\(
\phi(\vec{x} \mid \vec{\mu}, \vec{\Sigma})
= (2\pi)^{-k/2} |\vec{\Sigma}|^{-1/2}
\exp\{-\frac{1}{2} (\vec{x} - \vec{\mu})^\top \vec{\Sigma}^{-1} (\vec{x} - \vec{\mu})\}
\)
for $\vec{x} \in \mathbb{R}^k$, where the dimension $k$ depends on the context. Let $\text{IG}(a, b)$ denote the Inverse Gamma distribution with density $f_{\text{IG}}(x \mid a, b) = b^a x^{-a-1} e^{-b / x} / \Gamma(a) \cdot I(x > 0)$, where $I(\cdot)$ is the indicator function. First, the \textit{data model} is
\begin{align*}
Z_t^{(\ell)}(A) = Y_t^{(\ell)}(A) + \varepsilon_t^{(\ell)}(A), \quad
\varepsilon_t^{(\ell)}(A) \indep \text{N}(0, V_t^{(\ell)}(A)),
\end{align*}
for $A \in D_{t\ell}$ and $(v_t,\ell) \in (\mathcal{T} \times \mathcal{L})^*$. Second, the \textit{process model} is
\begin{align*}
Y_t^{(\ell)}(A) &=
\vec{h}(A)^\top \vec{\mu}_B + 
\vec{s}_t^{(\ell)}(A)^\top \vec{\eta} + 
\xi_t^{(\ell)}(A), \\
[\vec{\eta} \mid \sigma_K^2]  &\sim \text{N}(\vec{0}, \sigma_K^2 \vec{K}), \\
[\xi_t^{(\ell)}(A)\mid \sigma_{\xi}^2] &\iid \text{N}(0, \sigma_{\xi}^2),
\end{align*}
for $A \in D_{t\ell}$ and $(v_t,\ell) \in (\mathcal{T} \times \mathcal{L})^*$. Finally, the \textit{parameter model} is
\begin{align*}
\vec{\mu}_B \sim \text{N}(\vec{0}, \sigma_\mu^2 \vec{I}), \quad
\sigma_\mu^2 \sim \text{IG}(a_\mu, b_\mu), \quad
\sigma_K^2 \sim \text{IG}(a_K, b_K), \quad
\sigma_\xi^2 \sim \text{IG}(a_\xi, b_\xi).
\end{align*}
The STCOS model assumes that direct estimates $Z_t^{(\ell)}(A)$ constitute a noisy observation of an underlying latent process $Y_t^{(\ell)}(A)$. The variance of the noise $\varepsilon_t^{(\ell)}(A)$ is assumed to be the direct variance estimate $V_t^{(\ell)}(A)$. The mean of the latent process $Y_t^{(\ell)}(A)$ consists of a coarse spatial trend $\vec{h}(A)^\top \vec{\mu}_B$ and a spatio-temporal random process $\vec{s}_t^{(\ell)}(A)^\top \vec{\eta}$. Conjugate priors are assumed for the coefficients and variance parameters from the previous two stages. The matrix $\vec{K}$, which provides the covariance structure for the random coefficient of $\vec{\eta}$, is assumed to be known and is computable from the fine-level support.

The latent process model is motivated by the following construction. Define a continuous-space discrete-time process,
\begin{align*}
Y(\vec{u},v) = \delta(\vec{u}) + \sum_{j=1}^\infty \psi_j(\vec{u},v) \cdot \eta_j,
\quad \text{for $\vec{u} \in \bigcup_{i=1}^{n_B} B_i$ and $v \in \mathcal{T}$},
\end{align*}
where $\delta(\vec{u})$ is a large-scale spatial trend process and $\{ \psi_j(\vec{u},v) \}_{j=1}^\infty$ is a prespecified set of spatio-temporal basis functions. Integrating $Y(\vec{u},v)$ uniformly over $\vec{u} \in A$ and an $\ell$-year period $\vec{v} = (v_{t-\ell+1}, \ldots, v_t)$,
\begin{align}
Y_t^{(\ell)}(A) &=
\frac{1}{|A|} \int_A \delta(\vec{u}) \, d\vec{u} +
\frac{1}{\ell |A|} \sum_{k=t-\ell+1}^t \sum_{j=1}^r \int_A \psi_j(\vec{u},v_k) \cdot \eta_j \, d\vec{u} \nonumber \\
&\quad+ \frac{1}{\ell |A|} \sum_{k=t-\ell+1}^t \sum_{j=r+1}^\infty \int_A \psi_j(\vec{u},v_k) \cdot \eta_j \, d\vec{u} \nonumber \\
&= \mu(A) + \vec{s}_t^{(\ell)}(A)^\top \vec{\eta} + \xi_t^{(\ell)}(A).
\label{eqn:process}
\end{align}
In \eqref{eqn:process}, we have used the notation
\begin{align}
&\mu(A) = \frac{1}{|A|} \int_A \delta(\vec{u}) \, d\vec{u},
\label{eqn:large-scale} \\
&\vec{s}_t^{(\ell)}(A)^\top \vec{\eta} = \frac{1}{\ell |A|} \sum_{k=t-\ell+1}^t \sum_{j=1}^r \int_A \psi_j(\vec{u},v_k) \cdot \eta_j \, d\vec{u},
\label{eqn:spt-process} \\
&\xi_t^{(\ell)}(A) = \frac{1}{\ell |A|} \sum_{k=t-\ell+1}^t \sum_{j=r+1}^\infty \int_A \psi_j(\vec{u},v_k) \cdot \eta_j \, d\vec{u},
\label{eqn:remainder}
\end{align}
so that \eqref{eqn:large-scale} represents a large-scale spatial trend, \eqref{eqn:spt-process} is a spatio-temporal random process, and \eqref{eqn:remainder} is the remainder. We assume that $[\xi_t^{(\ell)}(A) \mid \sigma_\xi^2] \iid \text{N}(0, \sigma_\xi^2)$, and make use of local bisquare basis functions for the small-scale spatio-temporal trend, which are of the form
\begin{align}
\psi_j(\vec{u},v) =
\left[ 2 - \frac{\lVert \vec{u} - \vec{c}_j \rVert^2}{w_s^2}- \frac{|v - g_j|^2}{w_t^2} \right]^2  \cdot
I(\lVert \vec{u} - \vec{c}_j \rVert \leq w_s) \cdot
I(|v - g_j| \leq w_t),
\label{eqn:spacetime-bisquare-basis}
\end{align}
for $j = 1, \ldots, r$, with specified knots $\{ (\vec{c}_j, g_j) : j = 1, \ldots, r \}$. Knots may be taken as the Cartesian product of a set of spatial knot points $\{ \vec{c}_a : a = 1, \ldots, r_\text{space} \}$ and a set of temporal knot points $\{ g_b : b = 1, \ldots, r_\text{time} \}$; however, this is not required in general. The basis functions also require specification of a spatial radius $w_s$ and temporal radius $w_t$. We will select evenly spaced temporal knot points and spatial knot points according to a space-filling design \citep{NychkaSaltzman1998}. It can be difficult to specify $w_s$ directly, as the influence of $w_s$ depends on the coordinate system used in the supports. We therefore take $w_s = \tilde{w}_s \cdot Q_{0.05}$, where $Q_{0.05}$ is the 0.05 quantile of all nonzero pairwise distances between spatial cutpoints and $\tilde{w}_s$ is a parameter to be selected by the user. See \citet{JSM2017-STCOS} for a model selection study varying several factors in this model such as the number of knot points and the selection of $\tilde{w}_s$ and $w_t$.

Basis functions at the area level may be obtained from bases \eqref{eqn:spacetime-bisquare-basis} defined at the point level; for area $A$ and an $\ell$-year period specified by years $\vec{v} = (v_{t-\ell+1}, \ldots, v_t)$, let
\begin{align}
\psi_j^{(\ell)}(A, \vec{v})
= \frac{1}{\ell} \sum_{k=t-\ell+1}^t \frac{1}{|A|} \int_A \psi_j(\vec{u},v_k) d\vec{u},
\label{eqn:basis-spt}
\end{align}
which can be computed by a Monte Carlo approximation via
\begin{align*}
\psi_j^{(\ell)}(A, \vec{v})
\approx \frac{1}{\ell Q} \sum_{k=t-\ell+1}^t \sum_{q=1}^Q \psi_j(\vec{u}_q,v_k),
\end{align*}
based on a random sample of locations $\vec{u}_1, \ldots, \vec{u}_Q$ from a uniform distribution on the region $A$. Therefore, the basis expansion for an $\ell$-year lookback period $\vec{v} = (v_{t-\ell+1}, \ldots, v_t)$ and area $A$ is
\begin{align*}
\vec{s}_t^{(\ell)}(A)^\top = \left(
\psi_1^{(\ell)}(A, \vec{v}), \ldots, \psi_r^{(\ell)}(A, \vec{v})
\right).
\end{align*}
Next, for the large-scale spatial trend process, we make the simplifying assumption that
\begin{align*}
\delta(\vec{u}) = \sum_{i=1}^{n_B} \mu_i I(\vec{u} \in A \cap B_i),
\end{align*}
for an area $A$. Then $\delta(\vec{u})$ takes on a constant value on each overlap $A \cap B_i$ for $B_i \in D_B$. Define
\begin{align*}
\vec{h}(A) = \left( \frac{|A \cap B_1|}{|A|}, \ldots, \frac{|A \cap B_{n_B}|}{|A|} \right)^\top
\end{align*}
as the vector of proportions in which $A$ overlaps with each area $B_i$ in the fine-level support; this is based on the geography, and is therefore a known quantity in the analysis. Manipulation of geographical data in \proglang{R} will be discussed in Sections~\ref{sec:Rtools} and \ref{sec:columbia}. Integrating over $\vec{u} \in A$ yields
\begin{align*}
\mu(A) &= \frac{1}{|A|} \sum_{i=1}^{n_B} \int_{A \cap B_i} \delta(\vec{u}) d\vec{u}
= \frac{1}{|A|} \sum_{i=1}^{n_B} \mu_i \int_{A \cap B_i} du
= \sum_{i=1}^{n_B} \mu_i \frac{|A \cap B_i|}{|A|}
= \vec{h}(A)^\top \vec{\mu}_B.
\end{align*}
The coefficient
$
\vec{\mu}_B = (\mu_1, \ldots, \mu_{n_B})^\top
$
represents the change of support coefficient between the fine-level support and all other supports, and is the primary quantity of interest in the model.

To simplify the remaining presentation, we now write the model in vector form. Suppose there are $N$ total observations, indexed $i = 1, \ldots, N$, in all of the source supports combined. Let $\mathcal{H}$ be the mapping from each index $i$ to a triple $(A, t, \ell)$ consisting of the area $A$, time $v_t$, and lookback $\ell$ for the $i$th observation. Let $\text{vec}(\mathcal{S})$ denote a vector constructed from the elements of an ordered collection $\mathcal{S}$, $\Diag(\mathcal{S})$ represent a diagonal matrix with the elements of $\mathcal{S}$, and $\text{rbind}(\mathcal{S})$ represent a matrix with the elements of $\mathcal{S}$ as rows. We may then write
\begin{align}
\vec{Z} &= \text{vec}\left(
Z_t^{(\ell)}(A):
(A, t, \ell) = \mathcal{H}(i), \; i = 1, \ldots, N
\right), \nonumber \\
\quad \vec{H} &= \text{rbind}\left(
\vec{h}_t^{(\ell)}(A)^\top:
(A, t, \ell) = \mathcal{H}(i), \; i = 1, \ldots, N
\right), \nonumber \\
\vec{S} &= \text{rbind}\left(
\vec{s}_t^{(\ell)}(A)^\top:
(A, t, \ell) = \mathcal{H}(i), \; i = 1, \ldots, N
\right), \nonumber \\
\quad \vec{\xi} &= \text{vec}\left(
\xi_t^{(\ell)}(A):
(A, t, \ell) = \mathcal{H}(i), \; i = 1, \ldots, N
\right), \nonumber \\
\vec{\varepsilon} &= \text{vec}\left(
\varepsilon_t^{(\ell)}(A):
(A, t, \ell) = \mathcal{H}(i), \; i = 1, \ldots, N
\right), \nonumber \\
\quad \vec{V} &= \Diag\left(
V_t^{(\ell)}(A):
(A, t, \ell) = \mathcal{H}(i), \; i = 1, \ldots, N
\right),
\label{eqn:vector-form}
\end{align}
where $\vec{h}_t^{(\ell)}(A) = \vec{h}(A)$ does not vary with $t$ or $\ell$. The STCOS model can now be written
\begin{align}
&\vec{Z} = \vec{H} \vec{\mu}_B + \vec{S} \vec{\eta} + \vec{\xi} + \vec{\varepsilon}, \nonumber \\
&\vec{\varepsilon} \sim \text{N}(0, \vec{V}), \quad
[\vec{\eta} \mid \sigma_K^2] \sim \text{N}(\vec{0}, \sigma_K^2 \vec{K}), \quad
[\vec{\xi} \mid \sigma_\xi^2] \sim \text{N}(0, \sigma_{\xi}^2 \vec{I}), \nonumber \\
&[\vec{\mu}_B \mid \sigma_\mu^2] \sim \text{N}(\vec{0}, \sigma_\mu^2 \vec{I}), \quad
\sigma_\mu^2 \sim \text{IG}(a_\mu, b_\mu), \quad
\sigma_K^2 \sim \text{IG}(a_K, b_K), \quad
\sigma_\xi^2 \sim \text{IG}(a_\xi, b_\xi).
\label{eqn:stcos-vector}
\end{align}
We will also define $\vec{Y} = \vec{H} \vec{\mu}_B + \vec{S} \vec{\eta} + \vec{\xi}$ as the latent process for the observations.

We now discuss specification of the matrix $\vec{K}$. Let $A \sim B$ be the predicate that area $A$ is adjacent to area $B$ with $A \sim A$ taken to be false by definition. Denote $\vec{W} = (w_{ij})$ as the $n_B \times n_B$ adjacency matrix with $w_{ij} = I(B_i \sim B_j)$ for $i,j \in \{ 1, \ldots, n_B \}$, and $\vec{D} = \Diag(w_{1+}, \ldots, w_{n+})$ with $i$th diagonal entry $w_{i+} = \sum_{\ell=1}^{n_B} w_{ij}$. The matrix $\vec{Q} = \vec{I} - \tau \vec{D}^{-1} \vec{W}$ corresponds to the precision matrix of a particular class of conditional autoregressive (CAR) process. We take $\tau \in (0,1)$ to be known, for simplicity, to ensure that $\vec{Q}$ is nonsingular provided that areas $B_1, \ldots, B_{n_B}$ form a connected graph. Other choices of $\vec{Q}$ can be considered to obtain other classes of CAR precision matrices \citep[see][and the references therein]{CressieWikle2011,BanerjeeEtAl2014}.

For the purpose of specifying a spatio-temporal variance, suppose the fine-level support is distributed according to the process
\begin{align}
\vec{Y}_t^* = \vec{\mu}_B + \vec{\zeta}_t, \quad
\vec{\zeta}_t = \vec{M} \vec{\zeta}_{t-1} + \vec{b}_t, \quad
[\vec{b}_t \mid \sigma_K^2] \iid \text{N}(\vec{0}, \sigma_K^2 \vec{Q}^{-1}),
\label{eqn:fine-level}
\end{align}
for $v_t \in \mathcal{T}$ and assume $\vec{b}_0 = \vec{0}$. That is, $\{ \vec{Y}_t^* \}$ is a vector autoregressive (VAR) process in time and a CAR process in space. Let $\vec{\Sigma}_{Y^*}$ denote the covariance matrix of $(\vec{Y}_t^*: v_t \in \mathcal{T})$ under model \eqref{eqn:fine-level}. We take $\vec{K}$ to be the minimizer of
\begin{align}
\lVert \vec{\Sigma}_{Y^*} - \vec{S}^* \vec{C} \vec{S}^{*\top} \rVert_\text{F},
\quad \text{such that $\vec{C}$ is an $r \times r$ positive semidefinite matrix,}
\label{eqn:approximant-problem}
\end{align}
under the Frobenius norm $\lVert \cdot \rVert_\text{F}$, where
$\vec{S}^* = \text{rbind}\left( \vec{s}_t^{(\ell)}(A)^\top: A \in D_B, v_t \in \mathcal{T}, \ell \in \mathcal{L} \right)$
is the basis function expansion on the fine-level geography. In \eqref{eqn:approximant-problem}, $\vec{\Sigma}_{Y^*}$ represents the desired covariance structure under model \eqref{eqn:fine-level}, while $\vec{S}^* \vec{C} \vec{S}^{*\top}$ represents the realized covariance contribution of $\vec{S} \vec{\eta}$ in the model \eqref{eqn:stcos-vector}, where
\begin{align*}
\Var(\vec{Y} \mid \vec{\mu}_B, \sigma_\mu^2, \sigma_\xi^2, \sigma_K^2) = \sigma_K^2 \vec{S} \vec{K} \vec{S}^\top + \sigma_\xi^2 \vec{I},
\end{align*}
conditionally on the random variables in the parameter model. The solution to \eqref{eqn:approximant-problem},
\begin{align*}
\vec{C}^* = (\vec{S}^{*\top} \vec{S}^*)^{-1} \vec{S}^{*\top} \vec{\Sigma}_{Y^*} \vec{S}^* (\vec{S}^{*\top} \vec{S}^*)^{-1},
\end{align*}
provides the best positive approximant to $\vec{\Sigma}_{Y^*}$; details are given in Appendix~\ref{sec:details}. For the remainder of the article, we will take $\vec{\Sigma}_{Y^*}$ to be positive definite and $\vec{S}^*$ to be full rank so that $\vec{K}$ is positive definite. \citet{BradleyWikleHolan2015Stat} and \citet{BradleyHolanWikle2015AOAS} further discuss this approach within the context spatio-temporal models, and \citet{Higham1988} discusses the positive approximant problem in the general setting. We may write $\vec{\Sigma}_{Y^*} = \sigma_K^2 \tilde{\vec{\Sigma}}_{Y^*}$ so that
\begin{align}
\vec{C}^* = \sigma_K^2 \vec{K}, \quad
\vec{K} = (\vec{S}^{*\top} \vec{S}^*)^{-1} \vec{S}^{*\top} \tilde{\vec{\Sigma}}_{Y^*} \vec{S}^* (\vec{S}^{*\top} \vec{S}^*)^{-1}.
\label{eqn:approximant-solution}
\end{align}
Notice that $\tilde{\vec{\Sigma}}_{Y^*}$ and $\vec{K}$ are free of unknown parameters so that the solution of \eqref{eqn:approximant-problem} does not need to be recomputed within MCMC iterations as parameter values are updated.

We consider several possible structures for $\vec{K}$. First, assume that $\vec{M} = \vec{I}$ so that the fine-level process defined in \eqref{eqn:fine-level} is a vector random walk with nonstationary autocovariance function
\begin{align*}
\vec{\Gamma}(s,t) = \Cov(\vec{Y}_s^*, \vec{Y}_t^*)= \min(s,t) \sigma_K^2 \vec{Q}^{-1},
\end{align*}
conditioning on $\sigma_K^2$. Letting $\tilde{\vec{\Gamma}}(s,t) = \sigma_K^{-2} \vec{\Gamma}(s,t)$, which is free of $\sigma_K^{2}$, and choosing
\begin{align*}
\tilde{\vec{\Sigma}}_{Y^*} =
\begin{bmatrix}
\tilde{\vec{\Gamma}}(1,1) & \cdots & \tilde{\vec{\Gamma}}(1,T) \\
\vdots              & \ddots & \vdots \\
\tilde{\vec{\Gamma}}(T,1) & \cdots & \tilde{\vec{\Gamma}}(T,T)
\end{bmatrix}
\end{align*}
as the covariance of $\{ \vec{Y}_t^* \}$, $\vec{K}$ is obtained from \eqref{eqn:approximant-solution} to be
\begin{align}
\vec{K} &= (\vec{S}^{*\top} \vec{S}^*)^{-1} \left[
\sum_{s=1}^T \sum_{t=1}^T \min(s,t) \vec{S}_s^{*\top} \vec{Q}^{-1} \vec{S}_t^*
\right] (\vec{S}^{*\top} \vec{S}^*)^{-1},
\label{eqn:random-walk}
\end{align}
where we define
$\vec{S}_t^* = \text{rbind}\left( \vec{s}_t^{(\ell)}(A)^\top: A \in D_B, \ell \in \mathcal{L} \right)$
for each $v_t \in \mathcal{T}$. Another useful covariance structure arises if we assume that $\vec{M} = \vec{0}$. This yields autocovariance function $\vec{\Gamma}(s,t) = I(s = t) \sigma_K^2 \vec{Q}^{-1}$ and $\vec{\Sigma}_{Y^*} = \sigma_K^2 \vec{Q}^{-1} \otimes \vec{I}_{T}$, conditioning on $\sigma_K^2$, where $\otimes$ represents the Kronecker product. This structure supports nonzero covariance among areas at common times but independence between areas across times. The approximant \eqref{eqn:approximant-solution} with $\tilde{\vec{\Sigma}}_{Y^*} = \vec{Q}^{-1} \otimes \vec{I}_{T}$ is
\begin{align}
\vec{K} &= (\vec{S}^{*\top} \vec{S}^*)^{-1} \left[
\sum_{t=1}^T  \vec{S}_t^{*\top} \vec{Q}^{-1} \vec{S}_t^*
\right] (\vec{S}^{*\top} \vec{S}^*)^{-1}.
\label{eqn:spatial-only}
\end{align}
One more useful covariance structure assumes no spatial or temporal covariance;
\begin{align}
\vec{K} = \vec{I}.
\label{eqn:independence}
\end{align}
It is worth emphasizing that $\vec{K}$ describes the covariance structure for $\vec{\eta}$, but the covariance contribution to the model occurs through $\vec{Y}$ via $\vec{S} \vec{K} \vec{S}^\top$. For example, an independence assumption for $\vec{\eta}$ yields $\vec{S} \vec{K} \vec{S}^\top = \vec{S} \vec{S}^\top$, which is not necessarily a diagonal matrix using the basis functions \eqref{eqn:basis-spt} or the dimension-reduced version discussed in Section~\ref{sec:columbia}. The covariance structures we consider in this work---namely \eqref{eqn:random-walk}, \eqref{eqn:spatial-only}, and \eqref{eqn:independence}---are a departure from \citet{BradleyWikleHolan2015Stat}, who recommend computing $\vec{M}$ itself by further basis function decomposition.

We can obtain a Gibbs sampler by considering the joint distribution of the random quantities in \eqref{eqn:stcos-vector},
\begin{align*}
&f(\vec{Z}, \vec{\eta}, \vec{\xi}, \vec{\mu}_B, \sigma_\mu^2, \sigma_K^2, \sigma_\xi^2) \\
&\quad=
\phi(\vec{Z} \mid \vec{H} \vec{\mu}_B + \vec{S} \vec{\eta} + \vec{\xi}, \vec{V}) \cdot
\phi(\vec{\xi} \mid \vec{0}, \sigma_\xi^2 \vec{I}) \cdot
\phi(\vec{\eta} \mid 0, \sigma_K^2 \vec{K}) \\
&\qquad \times \phi(\vec{\mu}_B \mid \vec{0}, \sigma_\mu^2 \vec{I}) \cdot
f_{\text{IG}}(\sigma_\mu^2 \mid a_\mu, b_\mu) \cdot
f_{\text{IG}}(\sigma_K^2 \mid a_K, b_K) \cdot
f_{\text{IG}}(\sigma_\xi^2 \mid a_\xi, b_\xi),
\end{align*}
and deriving the full conditional distribution of each unknown parameter \citep[e.g.,][Section~5.3]{BanerjeeEtAl2014}. Here, the derivation is routine and details have been omitted for brevity. The steps of the Gibbs sampler which result from the full conditionals of $\vec{\mu}_B$, $\vec{\eta}$, $\vec{\xi}$, $\sigma_\mu^2$, $\sigma_K^2$, and $\sigma_\xi^2$ are stated as Algorithm~\ref{alg:gibbs-sampler}. The notation $[\vec{X} \mid \rest]$ is used to denote the distribution of a given random variable $\vec{X}$ conditioned on all other random quantities.

\begin{algorithm}
\caption{Gibbs sampler steps for STCOS model.}
\label{alg:gibbs-sampler}
\begin{enumerate}
\item Draw $[\vec{\mu}_B \mid \rest] \sim \text{N}(\vec{\vartheta}_\mu, \vec{\Omega}_\mu^{-1})$, with
$\vec{\Omega}_\mu = \vec{H}^\top \vec{V}^{-1} \vec{H} + \sigma_\mu^{-2} \vec{I}$.
and
$\vec{\vartheta}_\mu = \vec{\Omega}_\mu^{-1} \vec{H}^\top \vec{V}^{-1} (\vec{Z} - \vec{S} \vec{\eta} - \vec{\xi})$.

\item Draw $[\vec{\eta} \mid \rest] \sim \text{N}(\vec{\vartheta}_\eta, \vec{\Omega}_\eta^{-1})$, with
$\vec{\Omega}_\eta = \vec{S}^\top \vec{V}^{-1} \vec{S} + \sigma_K^{-2} \vec{K}^{-1}$
and
$\vec{\vartheta}_\eta = \vec{\Omega}_\eta^{-1} \vec{S}^\top \vec{V}^{-1} (\vec{Z} - \vec{H} \vec{\mu}_B - \vec{\xi})$.

\item Draw $[\vec{\xi} \mid \rest] \sim \text{N}(\vec{\vartheta}_\xi, \vec{\Omega}_\xi^{-1})$, with
$\vec{\Omega}_\xi = \vec{V}^{-1} + \sigma_\xi^{-2} \vec{I}$
and
$\vec{\vartheta}_\xi = \vec{\Omega}_\xi \vec{V}^{-1} (\vec{Z} - \vec{H} \vec{\mu}_B - \vec{S} \vec{\eta})$.

\item Draw $[\sigma_\mu^2 \mid \rest] \sim \text{IG}(a_\mu^*, b_\mu^*)$, with $a_\mu^* = a_\mu + n_B / 2$ and $b_\mu^* = b_\mu + \vec{\mu}_B^\top \vec{\mu}_B / 2$.

\item Draw $[\sigma_K^2 \mid \rest] \sim \text{IG}(a_K^*, b_K^*)$, with $a_K^* = a_K + r / 2$ and $b_K^* = b_K + \vec{\eta}^\top \vec{K}^{-1} \vec{\eta} / 2$.

\item Draw $[\sigma_\xi^2 \mid \rest] \sim \text{IG}(a_\xi^*, b_\xi^*)$, with $a_\xi^* = a_\xi + N / 2$ and $b_\xi^* = b_\xi + \vec{\xi}^\top \vec{\xi} / 2$.
\end{enumerate}
\end{algorithm}

\section[Implementing STCOS in R]{Implementing STCOS in \proglang{R}}
\label{sec:Rtools}

STCOS modeling can be roughly separated into three phases: assembling published estimates and geospatial data into a usable form, preparing matrices and vectors needed to fit the model, and finally fitting the model and producing results. To read and manipulate geospatial data, we will highlight the \pkg{sf} package \citep{Pebesma2018}, which we find to be intuitive and comprehensive. For general data manipulation, such as filtering records and selecting columns from a table, we will make use of the \pkg{dplyr} package \citep{dplyr2020}. To produce high quality graphics, we use the \pkg{ggplot2} package \citep{Wickham2016}. The \pkg{dplyr} and \pkg{ggplot2} packages are especially convenient because of their compatibility with \pkg{sf} objects. The \code{tigris} package \citep{Walker2018} provides a convenient way to request geographical data from the Census Bureau Tiger/Line database within \proglang{R}. The \pkg{fields} package \citep{NychkaEtAl2017} can be used to select spatial knot points by a space-filling design. General purpose platforms for Bayesian computing, including \proglang{Stan} \citep{CarpenterEtAl2017}, \proglang{JAGS} \citep{DepaoliEtAl2016}, \proglang{BUGS} \citep{LunnEtAl2009}, and \proglang{Nimble} \citep{deValpineEtAl2017}, are accessible through an \proglang{R} interface. A major advantage of such platforms is that samplers can be programmed simply by specifying a model and providing the data. In contrast, the traditional Gibbs sampler approach may require derivation, programming, and testing for each new model. However, general purpose platforms may not be well-suited to certain classes of models or to very large datasets. We will illustrate the use of \proglang{Stan} via the \pkg{rstan} package \citep{rstan2020} in addition to the Gibbs sampler from Section~\ref{sec:model}.

Some aspects of implementing STCOS analysis in \proglang{R} can be laborious and prone to error. To reduce the burden, we introduce the \pkg{stcos} package. The \pkg{stcos} package provides several major capabilities including: functions to compute overlap matrix $\vec{H}$ and adjacency matrix $\vec{W}$, basis functions to compute $\vec{S}$, construction of covariance $\vec{K}$, maximum likelihood estimation for the STCOS model, and an STCOS Gibbs sampler. Basis functions discussed in Section~\ref{sec:model} will be demonstrated shortly. Internal basis function calculations are carried out in \proglang{C++}, for efficiency, via the \pkg{Rcpp} and \pkg{RcppArmadillo} packages \citep{Eddelbuettel2013, EddelbuettelSanderson2014}. Matrices such as $\vec{H}$ and $\vec{S}$ are likely to be sparse in many STCOS applications; we use the \pkg{Matrix} package \citep{BatesMaechler2019} to support operations on sparse matrices.

We will now give an overview of the major STCOS computations which will be needed in \proglang{R}. Section~\ref{sec:columbia} will provide a demonstration connecting these pieces into a complete analysis. The following packages are assumed to be loaded in all coding examples.

\begin{CodeOutput}
R> library("sf")
R> library("dplyr")
R> library("stcos")
\end{CodeOutput}
A natural way to encode geographical features in source, fine-level, and target supports is via \code{sf} objects. Data associated with the supports can be embedded into \code{sf} objects to facilitate model preparation and graphical display. Therefore, our preferred workflow will be to produce \code{sf} objects with direct and model-based estimates. An example of a prepared source support object is as follows.
\begin{CodeOutput}
R> head(acs5_2013, 3)
Simple feature collection with 3 features and 8 fields
geometry type:  POLYGON
dimension:      XY
bbox:           xmin: -10280140 ymin: 4712766 xmax: -10277220 ymax: 4714750
CRS:            EPSG:3857
         geoid state county  tract blockgroup DirectEst DirectMOE DirectVar
1 290190005001    29    019 000500          1      9970      3157   3683788
2 290190005002    29    019 000500          2     12083      7048  18360194
3 290190006001    29    019 000600          1    105156     16979 106553987
                        geometry
1 POLYGON ((-10278231 4713772...
2 POLYGON ((-10279369 4713339...
3 POLYGON ((-10280135 4712926...
\end{CodeOutput}
Note that we have manipulated this output and some subsequent outputs to ensure that they fit on the page. The \code{CRS}  descriptor specifies a geographical coordinate system for the data. A number of standard coordinate systems are used to express geographical data, each having its own benefits and drawbacks. Coordinates such as latitude and longitude used in the global positioning system (GPS) describe points on the surface of the globe. Map projections provide two-dimensional representations of a region, which are convenient in many applications but necessarily distort the geography in some way. Conformal projections are designed to preserve local shape and are considered suitable for smaller domains, but distort areas when applied to large regions. On the other hand, equal-area projections are designed to preserve areas over large regions. To apply STCOS and other spatial-temporal methods, the analyst must select an appropriate coordinate system. We have chosen the Web Mercator projection (EPSG:3857) which is utilized in a number of online mapping services and is a variation of the conformal Mercator projection \citep{BattersbyEtAl2014}. Further discussion and references on coordinate systems can be found in \citet[Chapter~4]{BivandEtAl2013} and \citet[Chapter~3]{WallerGotway2004}.

All source, target, and fine-level supports in an analysis should use a common coordinate system so that they are compatible; this is not a limitation, as the analyst may transform an \code{sf} object from its original coordinates using the \code{sf::st_transform} function. Furthermore, methods described in this article are suited toward coordinates in a map projection rather than a globe representation. For example, the Euclidean distance utilized in \eqref{eqn:basis-spt} does not take into account the curvature of the Earth. An analysis using spherical coordinates---which may be appropriate for a larger-scale domain---might instead consider a great circle distance. Now that the importance of the coordinate system has been emphasized, coordinates will be considered as raw numerical values for the remainder of the article.

The last four lines of the previous display show a table with nine fields, where each row corresponds to an area (county) in the file. The \code{geometry} field contains details about the county's geography, which we typically will not want to manipulate directly. The fields \code{STATE} and \code{COUNTY} represent Federal Information Processing Standards (FIPS) codes for the state and county respectively, and \code{GEO_ID} is an identifier which combines the two. The fields \code{DirectEst}, \code{DirectMOE}, and \code{DirectVar} represent direct ACS estimates of median household income and an associated estimate of margin of error and variance. Preparation of such an \code{sf} object from geographical data and direct estimates will be discussed in Section~\ref{sec:assemble}.

A function \code{overlap_matrix} is provided to compute the $\vec{H}$ matrix.

\begin{CodeOutput}
R> H = overlap_matrix(dom1, dom2, proportion = TRUE)
\end{CodeOutput}
Here, \code{dom1} and \code{dom2} are \code{sf} objects which describe domains of areal units. The result is an \code{nrow(dom1)} by \code{nrow(dom2)} matrix. If \code{proportion = FALSE}, the entries represent the amount of area for each overlap; otherwise rows are normalized to proportions which sum to 1.

The \pkg{stcos} package provides several variations of the local bisquare basis functions discussed in Section~\ref{sec:model}. The following functions operate on data where space is represented at the point-level.

\begin{CodeOutput}
R> S = spatial_bisquare(dom, knots, w_s)
R> S = spacetime_bisquare(dom, knots, w_s, w_t)
\end{CodeOutput}
The function \code{spacetime_bisquare} implements \eqref{eqn:spacetime-bisquare-basis} which uses information both in space and time, while \code{spatial_bisquare} implements a space-only version
\begin{align}
\varphi_j(\vec{u}) =
\left[ 1 - \lVert \vec{u} - \vec{c}_j \rVert^2 / w^2 \right]^2  \cdot
I(\lVert \vec{u} - \vec{c}_j \rVert \leq w).
\label{eqn:spatial-bisquare-basis}
\end{align}
The object \code{dom} may either be a numerical matrix or an object of type \code{sf} or \code{sfc} containing points. In both cases, the first two columns/coordinates represent the spatial coordinates and the third represents time, if applicable. The object \code{knots} provides knot points, and may similarly be specified as either a numerical matrix or a \code{sf} or \code{sfc} object containing points. Coordinates systems for the points in \code{knots} are expected to be compatible with those in \code{dom}. Two-dimensional points are expected in \code{spatial_bisquare}, where each represents a $\vec{c}_j$. Similarly, three-dimensional points are expected in \code{spacetime_bisquare}, so that each represents a $(\vec{c}_j, g_j)$. The variables \code{w_s} and \code{w_t} correspond to the spatial and temporal radius, respectively. 

The following functions operate on data where space is represented at an area-level.

\begin{CodeOutput}
R> S = areal_spatial_bisquare(dom, knots, w_s, control = NULL)
R> S = areal_spacetime_bisquare(dom, period, knots, w_s, w_t, control = NULL)
\end{CodeOutput}
The function \code{areal_spacetime_bisquare} implements \eqref{eqn:basis-spt}, while \code{areal_spatial_bisquare} computes a space-only version
\[
\bar{\varphi}_j(A) = \frac{1}{|A|}\int_A \varphi_j(\vec{u}) d\vec{u},
\]
based on \eqref{eqn:spatial-bisquare-basis}. Here, the object \code{dom} is of type \code{sf} or \code{sfc} and provides the geography for one or more areal units. The variable \code{period} is a numeric vector which represents time period $\vec{v} = (v_{t-\ell+1}, \ldots, v_t)$ used to evaluate \eqref{eqn:basis-spt}. For example, if \code{dom} represents ACS 5-year estimates for 2017, we will take \code{period = 2013:2017}. The arguments \code{knots}, \code{w_s}, and \code{w_t} are interpreted similarly as in the point-level functions. The optional \code{control} argument is a list in which some additional factors can be adjusted, such as the number of Monte Carlo repetitions to used in the approximation. The remainder of the demonstration focuses on the STCOS analysis; further details and examples for the basis functions can be found in the \pkg{stcos} manual. 

Several options were described in Section~\ref{sec:model} to compute the covariance matrix $\vec{K}$; the \pkg{stcos} package provides functions to assist with the computations.

\begin{CodeOutput}
R> K = cov_approx_blockdiag(Qinv, S_fine)
R> K = cov_approx_randwalk(Qinv, S_fine)
\end{CodeOutput}
Both calls produce an $r \times r$ matrix. The call to \code{cov_approx_randwalk} corresponds to the random walk structure in \eqref{eqn:random-walk}, while \code{cov_approx_blockdiag} corresponds to \eqref{eqn:spatial-only} which assumes independence across time. The structure in \eqref{eqn:independence} which represents independent and identically distributed elements of $\vec{\eta}$ can be achieved with \code{K = Diagonal(n = r)}. The arguments \code{Qinv}, and \code{S_fine} correspond to the matrices $\vec{Q}^{-1}$ and $\vec{S}^*$ described in Section~\ref{sec:model}.
Note that the function \code{car_precision} in \pkg{stcos} can be used to compute $\vec{Q}$ from adjacency matrix $\vec{W}$.
\begin{CodeOutput}
R> Q = car_precision(W, tau = 0.9, scale = TRUE)
\end{CodeOutput}
The matrix $\vec{I} - \tau \vec{D}^{-1} \vec{W}$ is returned if \code{scale = TRUE}; otherwise $\vec{D} - \tau \vec{W}$ is returned.

Although we focus on Bayesian analysis, a function to compute maximum likelihood estimates (MLEs) is provided.

\begin{CodeOutput}
R> out = mle_stcos(z, v, H, S, K, init = list(sig2K = 1, sig2xi = 1))
R> sig2K_hat = out$sig2K_hat,
R> sig2xi_hat = out$sig2xi_hat,
R> mu_hat = out$mu_hat
\end{CodeOutput}
Some details on MLE computation are given in Appendix~\ref{sec:details}. MLE computation is often much quicker than Bayesian computation, and may provide good starting values for an MCMC sampler. Here, \code{H}, \code{S}, and \code{K} are the matrices $\vec{H}$, $\vec{S}$, and $\vec{K}$ described in \eqref{eqn:vector-form}, while \code{z} represents the vector $\vec{Z}$ and \code{v} is the diagonal of the matrix $\vec{V}$. The Gibbs sampler described in Section~\ref{sec:model} can be invoked using the \code{gibbs_stcos} function.

\begin{CodeOutput}
R> out = gibbs_stcos(z, v, H, S, Kinv = solve(K),
+  R = 10000, report_period = 1000, burn = 1000, thin = 10,
+  init = init)
R> muB_mcmc = out$muB_hist
R> eta_mcmc = out$eta_hist
R> xi_mcmc = out$xi_hist
R> sig2mu_mcmc = out$sig2mu_hist
R> sig2xi_mcmc = out$sig2xi_hist
R> sig2K_mcmc = out$sig2K_hist
\end{CodeOutput}
Some helper functions are provided to process the output from the Gibbs sampler.

\begin{CodeInput}
print(out)
logLik(out)
DIC(out)
E_mcmc = fitted(out, H_new, S_new)
Y_mcmc = predict(out, H_new, S_new)
\end{CodeInput}
The \code{print} function displays a brief summary of results from the sampler, while \code{logLik} computes the log-likelihood for each saved draw and \code{DIC} computes the Deviance information criterion \citep{SpiegelhalterEtAl2002} using saved draws. Let $\tilde{\vec{H}}$ be an $\tilde{N} \times n$ overlap matrix and $\tilde{\vec{S}}$ be an $\tilde{N} \times r$ basis matrix computed from target supports of interest, and \code{H_new} and \code{S_new} denote their representations in code. Let $\tilde{\vec{Y}}$ denote the vector composed of the $\tilde{N}$ latent process variables
\begin{align*}
\tilde{Y}_t^{(\ell)}(A) = \tilde{\vec{h}}(A)^\top \vec{\mu}_B + \tilde{\vec{s}}_t^{(\ell)}(A)^\top \vec{\eta} + \tilde{\xi}_t^{(\ell)}(A)
\end{align*}
associated with matrices $\tilde{\vec{H}}$ and $\tilde{\vec{S}}$. The \code{fitted} function produces draws from the posterior distribution of the mean
\begin{align*}
\E(\tilde{\vec{Y}} \mid \vec{\mu}_B, \vec{\eta}) = \tilde{\vec{H}} \vec{\mu}_B + \tilde{\vec{S}} \vec{\eta},
\end{align*}
so that \code{E_mcmc} is a matrix with $\tilde{N}$ columns where each row corresponds to a saved draw. Alternatively, the \code{predict} function produces draws from the posterior distribution of
\begin{align*}
\int \phi\left(\tilde{\vec{Y}} \mid \tilde{\vec{H}} \vec{\mu}_B + \tilde{\vec{S}} \vec{\eta}, \sigma_\xi^2 \vec{I} \right)
f(\vec{\mu}_B, \vec{\eta}, \sigma_\xi^2, \mid \vec{Z}, \vec{V}) \,
d\vec{\mu}_B \, d\vec{\eta} \, d\sigma_\xi^2.
\end{align*}

\section{Demonstration: City of Columbia neighborhoods}
\label{sec:columbia}
We now demonstrate an STCOS analysis on a small-scale but complete example using real data. Our target support consists of four neighborhoods in the City of Columbia in Boone County, Missouri. Geospatial data of the four neighborhoods has been provided by staff from the GIS Office for the City of Columbia. We would like to produce model-based estimates of median household income using observed ACS estimates from recent years. Specifically, we will consider 5-year ACS estimates at the block-group level for years 2013--2017 as our source supports, and will produce 5-year ACS estimates for year 2017 on the four neighborhoods as our target support.

The demonstration is split into several subsections. Section~\ref{sec:assemble} considers raw inputs---ACS direct estimates and geographical features---and discusses how they can be assembled into a useful form for the analysis. Section~\ref{sec:prepare} then prepares the inputs to the STCOS model: namely, $\vec{Z}$, $\vec{V}$, $\vec{H}$, $\vec{S}$, and $\vec{K}$. Section~\ref{sec:gibbs} uses the Gibbs sampler in the \pkg{stcos} package to produce draws from the posterior distribution of STCOS parameters and consequently obtain the desired results from the analysis. Section~\ref{sec:stan} uses the \proglang{Stan} platform via the \pkg{rstan} package as an alternative method to obtain results. Finally, Section~\ref{sec:mle} compares the MLE to Bayesian results.

\subsection{Assembling the data}
\label{sec:assemble}
We now briefly discuss how to gather geospatial data and ACS estimates and assemble them into \code{sf} objects for convenience. This is not intended to be an extensive guide, as numerous options to gather data (e.g.~portals, APIs, and \proglang{R} packages) are available and continue to evolve. In Section~\ref{sec:prepare}, we will make use of datasets which have been constructed for the demonstration.

Geospatial data representing the target support were provided in the shapefile format \citep{ESRI1998}. We now read the file and transform it to a projection of choice.

\begin{CodeOutput}
R> neighbs = st_read("neighborhoods.shp") %>% st_transform(crs = 3857)
\end{CodeOutput}
To prepare the source supports, we must gather ACS estimates and corresponding geographical features. For this example, ACS estimates can be requested from the Census Bureau's Data API. The interface and data availability of the API are subject to change in the future, and examples shown next may need to be modified accordingly. \citet{CensusAPI2019} provide a user guide with current specifications, including URL query format, available datasets, and codes for variable names. Note that limits are placed on the frequency and size of queries for unregistered users; higher-volume users may register for an API key to reduce restrictions. Estimates for our source supports can be requested from the API by constructing URLs with the following formats.

\begin{CodeOutput}
R> est_url = paste('https://api.census.gov/data/', year,
+  '/acs/acs5?get=NAME,B19013_001E&for=block%20group:*&in=state:29+county:019',
+  sep = '')
R> moe_url = paste('https://api.census.gov/data/', year,
+  '/acs/acs5?get=NAME,B19013_001M&for=block%20group:*&in=state:29+county:019',
+  sep = '')
\end{CodeOutput}
Data for the direct point estimates and corresponding MOEs have been gathered using two separate calls to the API. The FIPS code for Missouri is \code{29} and the code for Boone County is \code{019}. The variable \code{B19013_001E} represents point estimates for ``Median household income in the past 12 months'', and \code{B19013_001M} represents corresponding MOEs. We can request the years of interest by taking \code{year} to be values \code{2013} through \code{2017}. We use the \code{jsonlite} package \citep{Ooms2014} to call the API and load the results into an \proglang{R} \code{data.frame}.

\begin{CodeOutput}
R> json_data = jsonlite::fromJSON(est_url)
R> est_dat = data.frame(json_data[-1,])
R> colnames(est_dat) = json_data[1,]

R> json_data = jsonlite::fromJSON(moe_url)
R> moe_dat = data.frame(json_data[-1,])
R> colnames(moe_dat) = json_data[1,]
\end{CodeOutput}
We now merge \code{est\_dat} and \code{moe\_dat} together into a single \code{data.frame}.

\begin{CodeInput}
my_dat = est_dat %>%
	inner_join(moe_dat, by = c('state' = 'state', 'county' = 'county',
		'tract' = 'tract', 'block group' = 'block group')) %>%
	select(state, county, tract, blockgroup = `block group`,
		DirectEst = B19013_001E, DirectMOE = B19013_001M) %>%
	mutate(DirectEst = as.numeric(DirectEst)) %>%
	mutate(DirectMOE = as.numeric(DirectMOE)) %>%
	mutate(DirectEst = replace(DirectEst, DirectEst < 0, NA)) %>%
	mutate(DirectMOE = replace(DirectMOE, DirectMOE < 0, NA)) %>%	
	mutate(DirectVar = (DirectMOE / qnorm(0.95))^2) %>%
	arrange(tract, blockgroup)
\end{CodeInput}
There are a few details to mention in this data manipulation. We have taken some care because there is a space in the variable name \code{block group}, and because variables in the ACS data are interpreted as strings by default. We have transformed the MOE to a variance estimate, noting that published MOEs are to be interpreted as margins of error from $\alpha = 0.90$ confidence intervals \citep{ACSUsage2018};
i.e.,
\begin{align*}
\text{MOE} = z_{\alpha/2} \sqrt{\hat{\text{V}}}
\quad \iff \quad
\hat{\text{V}} = \left( \frac{\text{MOE}}{z_{\alpha/2}} \right)^2,
\end{align*}
where $z_{\alpha/2} \approx 1.645$. We have also taken care to handle special values coded in the data; namely, large negative numbers for estimates and MOEs are returned by the API when estimates are not available,%
\footnote{\url{https://census.gov/data/developers/data-sets/acs-1year/notes-on-acs-estimate-and-annotation-values.html}}
which we convert to \code{NA}. We sort the entries by tract and block group for readability. The resulting \code{data.frame} appears as follows.

\begin{CodeOutput}
R> head(my_dat)
  state county  tract blockgroup DirectEst DirectMOE DirectVar
1    29    019 000200          1     41063      6512  15673799
2    29    019 000200          2     31250      6978  17997303
3    29    019 000300          1     19420      7643  21591022
4    29    019 000300          2        NA        NA        NA
5    29    019 000300          3     21369     14558  78333750
6    29    019 000500          1     10995      5563  11438356
\end{CodeOutput}
The presence of \code{NA} values in direct estimates---such as in tract 000300, blockgroup 2---can vary over area, year, and period. \code{NA} values will be addressed in Section~\ref{sec:prepare}, before the analysis. The \code{tigris} package \citep{Walker2018} provides a convenient way to request shapefiles from the Census Bureau Tiger/Line database. It is necessary that all supports are converted to a common coordinate system for the analysis, so use the function \code{st::transform} to match the projection we used earlier in the target support.

\begin{CodeInput}
my_shp = tigris::block_groups(state = '29', county = '019', year = 2017) %>% 
	st_as_sf() %>%
	st_transform(crs = 3857)
\end{CodeInput}
Now we augment the geospatial data with direct point estimates, MOEs, and variance estimates obtained earlier.

\begin{CodeInput}
acs5_2017 = my_shp %>%
	inner_join(my_dat, by = c('STATEFP' = 'state', 'COUNTYFP' = 'county',
		'TRACTCE' = 'tract', 'BLKGRPCE' = 'blockgroup')) %>%
	select(geoid = GEOID, state = STATEFP, county = COUNTYFP,
		tract = TRACTCE, blockgroup = BLKGRPCE,
		DirectEst, DirectMOE, DirectVar)
\end{CodeInput}
The resulting \code{acs5_2017} is an object of type \code{sf}, whose first few entries are as follows (\code{geometry} column is not shown).

\begin{CodeOutput}
R> head(acs5_2017)
Simple feature collection with 6 features and 8 fields
geometry type:  POLYGON
dimension:      XY
bbox:           xmin: -10280690 ymin: 4712766 xmax: -10256290 ymax: 4752109
CRS:            EPSG:3857
         geoid state county  tract blockgroup DirectEst DirectMOE DirectVar
1 290190005001    29    019 000500          1     10995      5563  11438356
2 290190005002    29    019 000500          2     13872      9503  33378510
3 290190006001    29    019 000600          1     45208     39073 564285643
4 290190006002    29    019 000600          2    107500     19868 145899495
5 290190020002    29    019 002000          2     62237     13529  67651414
6 290190020003    29    019 002000          3     51019     11166  46082999
\end{CodeOutput}

\subsection{Preparing the analysis}
\label{sec:prepare}
The steps in Section~\ref{sec:assemble} can be repeated so that all target, source, and fine-level supports are assembled as \code{sf} objects. The \pkg{stcos} package includes the following pre-constructed datasets to facilitate our demonstration.

\begin{CodeOutput}
R> data("acs_sf")
R> ls(pattern = "acs5_.*")
[1] "acs5_2013" "acs5_2014" "acs5_2015" "acs5_2016" "acs5_2017"
R> data("columbia_neighbs")
R> ls(pattern = "columbia")
[1] "columbia_neighbs"	
\end{CodeOutput}
Before we begin to prepare the terms in \eqref{eqn:vector-form} for the STCOS model, let us create a version of the source supports with \code{NA} estimates removed. This will help to avoid complications in model fitting.

\begin{CodeInput}
source_2013 = acs5_2013 %>% filter(!is.na(DirectEst) & !is.na(DirectVar))
source_2014 = acs5_2014 %>% filter(!is.na(DirectEst) & !is.na(DirectVar))
source_2015 = acs5_2015 %>% filter(!is.na(DirectEst) & !is.na(DirectVar))
source_2016 = acs5_2016 %>% filter(!is.na(DirectEst) & !is.na(DirectVar))
source_2017 = acs5_2017 %>% filter(!is.na(DirectEst) & !is.na(DirectVar))
\end{CodeInput}
We will choose our fine-level support based on the \code{acs5_2017} geography; i.e.~the block group level geography for Boone County in 2017. However, because we have dropped some areas from the source supports, we should check for areas in \code{acs5_2017} which have zero or very little overlap with any areas in the source supports. If we identify such areas, we will drop them from the analysis to avoid rank-deficiency of the $\vec{H}$ matrix.

\begin{CodeInput}
U = rbind(
	overlap_matrix(source_2013, acs5_2017, proportion = FALSE),
	overlap_matrix(source_2014, acs5_2017, proportion = FALSE),
	overlap_matrix(source_2015, acs5_2017, proportion = FALSE),
	overlap_matrix(source_2016, acs5_2017, proportion = FALSE),
	overlap_matrix(source_2017, acs5_2017, proportion = FALSE)
)
dom_fine = acs5_2017 %>%
	mutate(keep = (colSums(U) >= 10)) %>%
	filter(keep == TRUE) %>%
	select(-c("DirectEst", "DirectMOE", "DirectVar", "keep"))
n = nrow(dom_fine)
\end{CodeInput}
This creates \code{dom_fine} as a version of \code{acs5_2017}, excluding two block-groups having very little overlap (less than 10 square meters) with any of the source support areas, and ignoring the columns for the direct estimates, MOEs, and variance estimates.

The overlap matrix $\vec{H}$ for the analysis can now be created as follows.

\begin{CodeInput}
H = rbind(
	overlap_matrix(source_2013, dom_fine),
	overlap_matrix(source_2014, dom_fine),
	overlap_matrix(source_2015, dom_fine),
	overlap_matrix(source_2016, dom_fine),
	overlap_matrix(source_2017, dom_fine)
)
N = nrow(H)
\end{CodeInput}
To construct a bisquare basis, we must select spatio-temporal knot points. To select spatial knot points, we first draw a large number of points uniformly over the fine-level domain using the \code{st_sample} function. We then use the \code{cover.design} function in the \pkg{fields} package, which finds a subset of these points to fill the space.

\begin{CodeInput}
u = st_sample(dom_fine, size = 2000)
P = st_coordinates(u)
out = fields::cover.design(P, 200)
knots_sp = out$design
\end{CodeInput}
To select the spatial radius $w_s$, we compute the $0.05$ quantile of the pairwise distances among the rows of \code{knots_sp}, as discussed in Section~\ref{sec:model}.

\begin{CodeInput}
ws_tilde = 1
D = dist(knots_sp)
w_s =  ws_tilde * quantile(D[D > 0], prob = 0.05, type = 1)
\end{CodeInput}
Alternatively, evenly spaced points can be achieved with the hexagonal sampling method in the \code{sf::st_sample} function. This is quicker than \code{fields::cover.design}.

\begin{CodeInput}
u = st_sample(dom_fine, 200, type = "hexagonal")
knots_sp_alt = st_coordinates(u)
D = dist(knots_sp_alt)
w_s_alt = ws_tilde * quantile(D[D > 0], prob = 0.05, type = 1)
\end{CodeInput}
Figure~\ref{fig:spatial-knots} illustrates the selected spatial knot points and radius using both the space-filling method and hexagonal sampling. Both methods succeed in creating a grid of evenly-spaced points, although the latter follow a more strict pattern. More evenly-spaced points can also be obtained with the space-filling method by taking an initial sample size larger than our selection of 2,000.

We choose the temporal knot points to be $(2009, 2009.5, \ldots, 2016.5, 2017)$, covering the years relevant to the 5-year ACS estimates for years 2013--2017.

\begin{CodeInput}
knots_t = seq(2009, 2017, by = 0.5)
w_t = 1
\end{CodeInput}
More sophisticated date/time functions can assist in constructing temporal knots, though a numerical representation is ultimately needed. An alternative choice for temporal knots created with \code{Date} objects is given next. When treated as numerical, such objects represent days elapsed since January 1, 1970. Here we may again use the quantile approach to determine a radius which is suitable for this unit of time.

\begin{CodeInput}
dates = seq(as.Date("2009-01-01"), as.Date("2017-01-01"), by = "6 months")
knots_t_alt = as.numeric(dates)
wt_tilde_alt = 1
D = dist(knots_t_alt)
w_t_alt = wt_tilde_alt * quantile(D[D > 0], prob = 0.05, type = 1)
\end{CodeInput}
Now we use the \code{merge} function in the \pkg{base} package \citep{Rcore2020} to perform a Cartesian join between the spatial knots \code{knots_sp} and temporal knots \code{knots_t}, which yields the set of spatio-temporal knots.

\begin{CodeInput}
knots = merge(knots_sp, knots_t)
\end{CodeInput}
Now, we use the basis functions to compute the design matrix $\vec{S}$.

\begin{CodeInput}
bs_ctrl = list(mc_reps = 500)
S_full = rbind(
	areal_spacetime_bisquare(source_2013, 2009:2013, knots, w_s, w_t, bs_ctrl),
	areal_spacetime_bisquare(source_2014, 2010:2014, knots, w_s, w_t, bs_ctrl),
	areal_spacetime_bisquare(source_2015, 2011:2015, knots, w_s, w_t, bs_ctrl),
	areal_spacetime_bisquare(source_2016, 2012:2016, knots, w_s, w_t, bs_ctrl),
	areal_spacetime_bisquare(source_2017, 2013:2017, knots, w_s, w_t, bs_ctrl)
)
\end{CodeInput}
We can also compute the design matrix $\vec{S}^*$ on the fine-level support, which is needed to compute $\vec{K}$ under some of the possible structures.

\begin{CodeInput}
S_fine_full = rbind(
	areal_spacetime_bisquare(dom_fine, 2009, knots, w_s, w_t, bs_ctrl),
	areal_spacetime_bisquare(dom_fine, 2010, knots, w_s, w_t, bs_ctrl),
	areal_spacetime_bisquare(dom_fine, 2011, knots, w_s, w_t, bs_ctrl),
	areal_spacetime_bisquare(dom_fine, 2012, knots, w_s, w_t, bs_ctrl),
	areal_spacetime_bisquare(dom_fine, 2013, knots, w_s, w_t, bs_ctrl),
	areal_spacetime_bisquare(dom_fine, 2014, knots, w_s, w_t, bs_ctrl),
	areal_spacetime_bisquare(dom_fine, 2015, knots, w_s, w_t, bs_ctrl),
	areal_spacetime_bisquare(dom_fine, 2016, knots, w_s, w_t, bs_ctrl),
	areal_spacetime_bisquare(dom_fine, 2017, knots, w_s, w_t, bs_ctrl)
)
\end{CodeInput}
Next we need vectors \code{z} and \code{v} to represent the direct point estimates and associated variance estimates.

\begin{CodeInput}
z = c(source_2013$DirectEst, source_2014$DirectEst, source_2015$DirectEst,
	source_2016$DirectEst, source_2017$DirectEst)
v = c(source_2013$DirectVar, source_2014$DirectVar, source_2015$DirectVar,
	source_2016$DirectVar, source_2017$DirectVar)
\end{CodeInput}
Because \code{z} and \code{v} contain rather large numbers, we standardize \code{z} for the analysis and make a corresponding transformation to \code{v}.

\begin{CodeInput}
z_scaled = scale(z)
v_scaled = v / var(z)
\end{CodeInput}
The expression for \code{v_scaled} arises from considering $\Var[a^{-1/2}(Z_i - b)] = a^{-1} \Var(Z_i)$ for constants $a > 0$ and $b \in \mathbb{R}$, which is estimated by $a^{-1} \vec{e}_i^\top \vec{V} \vec{e}_i$ with $\vec{e}_i$ the $i$th column of an $N \times N$ identity matrix. The design matrix $\vec{S}$ with our choice of basis function can have a large number of columns and a high degree of multicollinearity; if not addressed, this can lead to poor mixing in the MCMC sampler. A simple workaround is to reduce the dimension of $\vec{S}$ using principal components analysis (PCA). First we compute the reduction, using 65\% of the variability, as expressed as a proportion of the eigenvalues.

\begin{CodeInput}
eig = eigen(t(S_full) %*% S_full)
idx_S = which(cumsum(eig$values) / sum(eig$values) < 0.65)
\end{CodeInput}
Figure~\ref{fig:pca-reduction} shows that this can be accomplished by projecting from the original 3,400 columns to $r=19$ columns. Now we apply the reduction to $\vec{S}$ as well as $\vec{S}^*$. 

\begin{CodeInput}
Tx_S = eig$vectors[,idx_S]
S = S_full %*% Tx_S
S_fine = S_fine_full %*% Tx_S
r = ncol(S)
\end{CodeInput}
The last ingredient needed to run the analysis is the matrix $\vec{K}$. We will use the random walk structure in \eqref{eqn:random-walk} to express both spatial and temporal dependence. First, let us compute the covariance matrix $\vec{Q}^{-1}$ of a CAR process for the fine-level support.

\begin{CodeInput}
W = adjacency_matrix(dom_fine)
Q = car_precision(W, tau = 0.9, scale = TRUE)
Qinv = solve(Q)
\end{CodeInput}
Now compute $\vec{K}$ using $\vec{Q}^{-1}$ and $\vec{S}^*$.

\begin{CodeInput}
K = cov_approx_randwalk(Qinv, S_fine)
\end{CodeInput}

\subsection{Fitting with Gibbs sampler}
\label{sec:gibbs}
We now proceed to run the Gibbs sampler. We will produce a chain of 10,000 iterations, discard the first 2,000 draws, and keep one of every 10th remaining draw. We will use hyperparameters $a_\mu = 1$, $b_\mu = 1$, $a_K = 1$, $b_K = 2$, $a_\xi = 1$, and $b_\xi = 2$.

\begin{CodeOutput}
R> hyper = list(a_sig2K = 1, b_sig2K = 2, a_sig2xi = 1, b_sig2xi = 2,
+  a_sig2mu = 1, b_sig2mu = 2)
R> gibbs_out = gibbs_stcos(z = z_scaled, v = v_scaled, H = H, S = S,
+  Kinv = Kinv, R = 10000, report_period = 2000, burn = 2000,
+  thin = 10, hyper = hyper)
2020-05-17 17:11:15 - Begin Gibbs sampler
2020-05-17 17:11:52 - Begin iteration 2000
2020-05-17 17:12:29 - Begin iteration 4000
2020-05-17 17:12:50 - Begin iteration 6000
2020-05-17 17:13:09 - Begin iteration 8000
2020-05-17 17:13:28 - Begin iteration 10000
2020-05-17 17:13:28 - Finished Gibbs sampler
R> print(gibbs_out)
Fit for STCOS model
--
             Mean          SD       2.5%        25%        75%      97.5%
sig2mu 0.52574414 0.094603198 0.36964018 0.45835116 0.57962109 0.74252701
sig2K  1.12632689 0.829211900 0.34238233 0.62142495 1.34628662 3.35183509
sig2xi 0.04368451 0.005022067 0.03485603 0.04021765 0.04669926 0.05533603
--
Saved 800 draws
DIC: 210.798981
Elapsed time: 00:02:08
\end{CodeOutput}
The \code{mcmc} class in the \pkg{coda} package \citep{PlummerEtAl2006} helps to manage and plot the draws.

\begin{CodeInput}
library("coda")
varcomps_mcmc = mcmc(data.frame(
	sig2mu = gibbs_out$sig2mu_hist,
	sig2xi = gibbs_out$sig2xi_hist,
	sig2K = gibbs_out$sig2K_hist
))
plot(varcomps_mcmc)
\end{CodeInput}
Figure~\ref{fig:trace-varcomps} displays trace and density plots of the variance components $\sigma_\mu^2$, $\sigma_\xi^2$, and $\sigma_K^2$.

Using the fitted model, we can produce model-based estimates on target supports of interest. In this example, we would like to produce 5-year 2017 estimates for our four neighborhoods in Boone County: Central, East, North, and Paris. The following code computes model-based estimates for these areas and embeds them into an \code{sf} object for plotting.

\begin{CodeInput}
nb_out = neighbs
H_new = overlap_matrix(nb_out, dom_fine)               # New overlap
S_new_full = areal_spacetime_bisquare(nb_out,
	2013:2017, knots, w_s, w_t, bs_ctrl)               # New basis fn
S_new = S_new_full %*% Tx_S                            # Reduce dimension

EY_scaled = fitted(gibbs_out, H_new, S_new)            # Get draws of E(Y)
EY = sd(z) * EY_scaled + mean(z)                       # Uncenter and unscale

alpha = 0.10
nb_out$E_mean = colMeans(EY)						   # Point estimates
nb_out$E_sd = apply(EY, 2, sd)						   # SDs
nb_out$E_lo = apply(EY, 2, quantile, prob = alpha/2)   # Credible interval lo
nb_out$E_hi = apply(EY, 2, quantile, prob = 1-alpha/2) # Credible interval hi
nb_out$E_median = apply(EY, 2, median)				   # Median
nb_out$E_moe = apply(EY, 2, sd) * qnorm(1-alpha/2)	   # MOE
\end{CodeInput}

The objects \code{H_new} and \code{S_new} represent design matrices $\tilde{\vec{H}}$ and $\tilde{\vec{S}}$, respectively, based on the geography of neighborhoods. The \code{fitted} function was then used to produce draws from the posterior distribution of $\E(\tilde{\vec{Y}})$. We then transformed the resulting estimates back to the original scale, having previously centered and scaled them before model fitting. The remainder of the code display summarizes draws of the posterior mean in several ways, obtaining a model-based estimate of its median, mean, standard deviation, MOE ($z_{\alpha/2} \times \text{standard deviation}$), and a 90\% credible interval. The resulting \code{sf} object is displayed below (\code{geometry} column is not shown).

\begin{CodeOutput}
R> print(nb_out)
Simple feature collection with 4 features and 7 fields
geometry type:  POLYGON
dimension:      XY
bbox:           xmin: -10280270 ymin: 4715036 xmax: -10269750 ymax: 4723860
CRS:            EPSG:3857
           Region   E_mean     E_sd     E_lo     E_hi E_median    E_moe
1         Central 26705.85 1921.623 23456.96 29709.53 26720.70 3160.788
2            East 44127.78 2450.811 40155.21 47983.46 44186.03 4031.225
3           North 44171.24 2863.373 39519.40 48933.11 44040.79 4709.829
4 Paris63Corridor 20386.72 3663.098 14448.61 26325.73 20309.78 6025.261
\end{CodeOutput}
We are now ready to plot our estimates. The code to reproduce our plots is somewhat lengthy and can be found in the supplemental materials. First we compare direct and model-based estimates for 2017 source supports to assess whether the model fit is reasonable. Figures~\ref{fig:compare2017-direct} and \ref{fig:compare2017-model} show maps of the two sets of estimates. Figures~\ref{fig:compare2014-scatter} and \ref{fig:compare2017-scatter} compare the two sets of estimates via scatter plots; year 2014 and year 2017 estimates are shown for comparison. Variation between direct and model-based estimates appears to be smaller for year 2014, with the exception of the block group with the largest direct estimate that year. Finally, Figure~\ref{fig:areas-of-interest} shows the four neighborhoods of our target support in the context of the 2017 5-year direct estimates. This provides a visual aid to assess plausibility of the target support estimates. The North and East neighborhoods appear to be in the immediate vicinity of block groups with higher median household income than the West and Paris neighborhoods.

\subsection[Fitting with Stan]{Fitting with \pkg{Stan}}
\label{sec:stan}
We will now refit the model from Section~\ref{sec:gibbs} using \pkg{Stan} instead of the Gibbs sampler. First, we will need a \pkg{Stan} model specification. We will create a file named \code{stcos.stan} with the following contents.

\begin{CodeInput}
data {
	int<lower=0> N;  int<lower=0> n;  int<lower=0> r;
	vector[N] z;     vector[N] v;     matrix[N,n] H;
	matrix[N,r] S;   matrix[r,r] K;   real alpha_K;
	real beta_K;     real alpha_xi;   real beta_xi;
	real alpha_mu;   real beta_mu;
}
parameters {
	vector[n] mu;    real<lower=0> sig2K;
	vector[r] eta;   real<lower=0> sig2xi;
	vector[N] xi;    real<lower=0> sig2mu;
}
model {
	sig2K ~ inv_gamma(alpha_K, beta_K);
	sig2xi ~ inv_gamma(alpha_xi, beta_xi);
	sig2mu ~ inv_gamma(alpha_mu, beta_mu);
	eta ~ multi_normal(rep_vector(0,r), sig2K * K);
	mu ~ normal(0, sqrt(sig2mu));
	xi ~ normal(0, sqrt(sig2xi));
	z ~ normal(H*mu + S*eta + xi, sqrt(v));
}
\end{CodeInput}
Now, in \code{R}, pass the data and model specification to \code{stan} to initiate fitting.

\begin{CodeInput}
library("rstan")
stan_dat = list(
	N = N, n = n, r = r, z = z_scaled, v = v_scaled, H = as.matrix(H),
	S = as.matrix(S), K = as.matrix(K),
	alpha_K = 1, beta_K = 2, alpha_xi = 1, beta_xi = 2, alpha_mu = 1, beta_mu = 2
)
\end{CodeInput}
\begin{CodeOutput}
R> stan_out = stan(file = "stcos.stan", data = stan_dat, iter = 2000, chains = 2)
SAMPLING FOR MODEL 'stcos' NOW (CHAIN 1).
...
Chain 1:  Elapsed Time: 11.8561 seconds (Warm-up)
Chain 1:                10.5266 seconds (Sampling)
Chain 1:                22.3827 seconds (Total)
...
SAMPLING FOR MODEL 'stcos' NOW (CHAIN 2).
...
Chain 2:  Elapsed Time: 10.9951 seconds (Warm-up)
Chain 2:                9.87796 seconds (Sampling)
Chain 2:                20.873 seconds (Total)
\end{CodeOutput}
Here we have requested two chains of length 2,000 each. In addition to the time needed for sampling, \pkg{Stan} may require time to compile the model specification. Upon successful completion of sampling, the following \code{R} code can be used to extract draws and produce results.

\begin{CodeInput}
stan_draws = extract(stan_out, pars = c("mu", "eta"), permuted = TRUE)

nb_out = neighbs
H_new = overlap_matrix(nb_out, dom_fine)               # New overlap
S_new_full = areal_spacetime_bisquare(nb_out,
	2013:2017, knots, w_s, w_t, bs_ctrl)               # New basis fn
S_new = S_new_full %*% Tx_S                            # Reduce dimension

EY_scaled = stan_draws$mu %*% t(H_new) +
    stan_draws$eta %*% t(S_new)                        # Draws of E(Y)
EY = sd(z) * EY_scaled + mean(z)                       # Uncenter and unscale

alpha = 0.10
nb_out$E_mean = colMeans(EY)                           # Point estimates
nb_out$E_sd = apply(EY, 2, sd)                         # SDs
nb_out$E_lo = apply(EY, 2, quantile, prob = alpha/2)   # Credible interval lo
nb_out$E_hi = apply(EY, 2, quantile, prob = 1-alpha/2) # Credible interval hi
nb_out$E_median = apply(EY, 2, median)                 # Median
nb_out$E_moe = apply(EY, 2, sd) * qnorm(1-alpha/2)     # MOE
\end{CodeInput}
The result of \code{print(nb_out)} can be compared to the corresponding output from the Gibbs sampler in Section~\ref{sec:gibbs}.

\subsection{Fitting with Maximum Likelihood}
\label{sec:mle}
Finally, we compute maximum likelihood estimates to compare to our Bayesian results.

\begin{CodeOutput}
R> mle_out = mle_stcos(z_scaled, v_scaled, H, S, K)
R> print(mle_out$sig2K_hat)
[1] 1.310004e-11
R> print(mle_out$sig2xi_hat)
[1] 1.225813e-11
R> print(mle_out$loglik)
[1] 67.29006
\end{CodeOutput}
Estimates for both $\sigma_\xi^2$ and $\sigma_K^2$ are very small, which indicates that the direct variance estimates $\vec{V}$ are capturing much of the variability among $\vec{Z}$. This can be contrasted with the Bayesian approach, which finds a non-zero effect of $\sigma_\xi^2$ and $\sigma_K^2$ through the addition of prior information. The following code extracts the MLE $\hat{\vec{\mu}}_B$, computes estimates $\tilde{\vec{H}} \hat{\vec{\mu}}$ of the four neighborhoods, and transforms those estimates to the original scale of the direct estimates.

\begin{CodeInput}
H_new = overlap_matrix(neighbs, dom_fine)
mu_hat = mle_out$mu_hat
z_hat_scaled = as.numeric(H_new %*% mu_hat)
z_hat = sd(z) * z_hat_scaled + mean(z)
\end{CodeInput}
Figure~\ref{fig:mle_vs_gibbs} plots MLEs with box plots of corresponding saved draws from the Gibbs sampler. Figure~\ref{fig:mle_vs_gibbs_mu_part1} plots the first 10 components of $\hat{\vec{\mu}}_B$, while Figure~\ref{fig:mle_vs_gibbs_neighb} plots the four components of $\tilde{\vec{H}} \hat{\vec{\mu}}$ corresponding to neighborhoods. Code to reproduce Figure~\ref{fig:mle_vs_gibbs} is provided in the supplemental materials. Bayesian and maximum likelihood estimates are not seen to be vastly different, and we anticipate that they would become closer as the total number of observations $N$ becomes large relative to $n_B$. In the current setting where $N = 421$ and $n_B = 85$, we would recommend the Bayesian approach.

\section{Conclusions}
\label{sec:conclusions}
In this article, we have demonstrated a complete implementation of STCOS methodology for \proglang{R} users. We worked through a small example to estimate median household income in several neighborhoods in the City of Columbia in Boone County, MO. Established \proglang{R} packages such as \pkg{sf}, \pkg{dplyr}, \pkg{Matrix}, and \pkg{rstan} were instrumental in the process, from initially gathering the data, to carrying out the MCMC, to placing results into a usable form. The \pkg{stcos} package was introduced to assist with some intricate programming steps not covered by other packages, especially computing areal spatio-temporal basis functions. Use of the highlighted tools significantly reduces the learning curve to program an analysis; however, some technical experience and effort are still required for a successful implementation. Future work may involve additional improvements to the \pkg{stcos} package for efficiency and usability, as well as software support for other spatial and spatio-temporal methodologies.

\section*{Acknowledgements}
This research was partially supported by the U.S. National Science Foundation (NSF) and the U.S. Census Bureau under NSF grant SES-1132031, funded through the NSF-Census Research Network (NCRN) program, and NSF Awards SES-1853096 and SES-1853099. This article is released to inform interested parties of ongoing research and to encourage discussion. The views expressed on statistical issues are those of the authors and not the NSF or U.S. Census Bureau. The authors thank Taylor Bowen and Toni Messina from the Office of Information Technology/GIS, City of Columbia, Missouri for supplying the shapefile used in the case study and for useful discussion.


\clearpage

\begin{figure}
\centering
\includegraphics[width=0.5\textwidth]{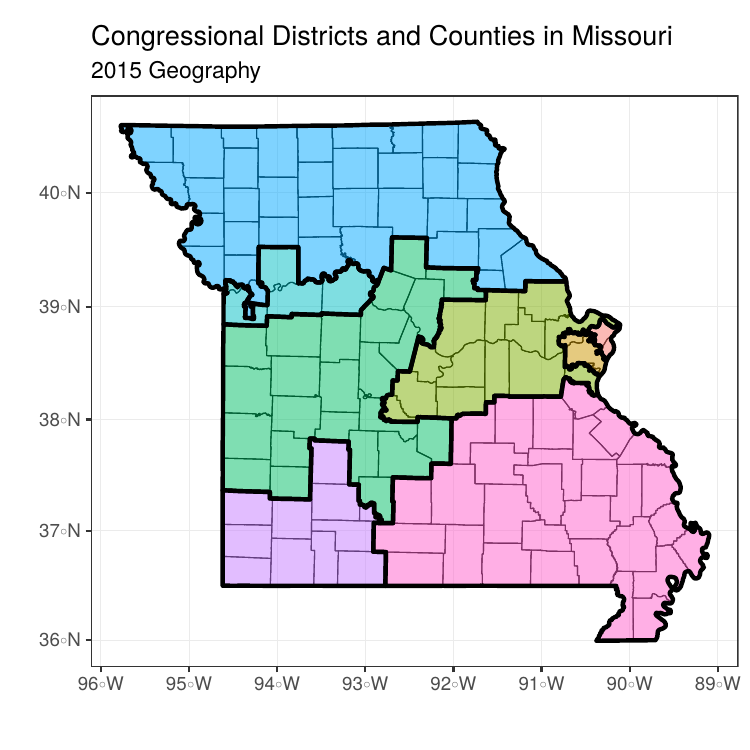}
\caption{The state of Missouri in 2015. Thin lines mark boundaries between the 114 counties and one independent city. Shaded areas with thick lines mark the eight congressional districts.}
\label{fig:county-vs-cd}
\end{figure}

\begin{figure}
\centering
\includegraphics[width=0.48\textwidth]{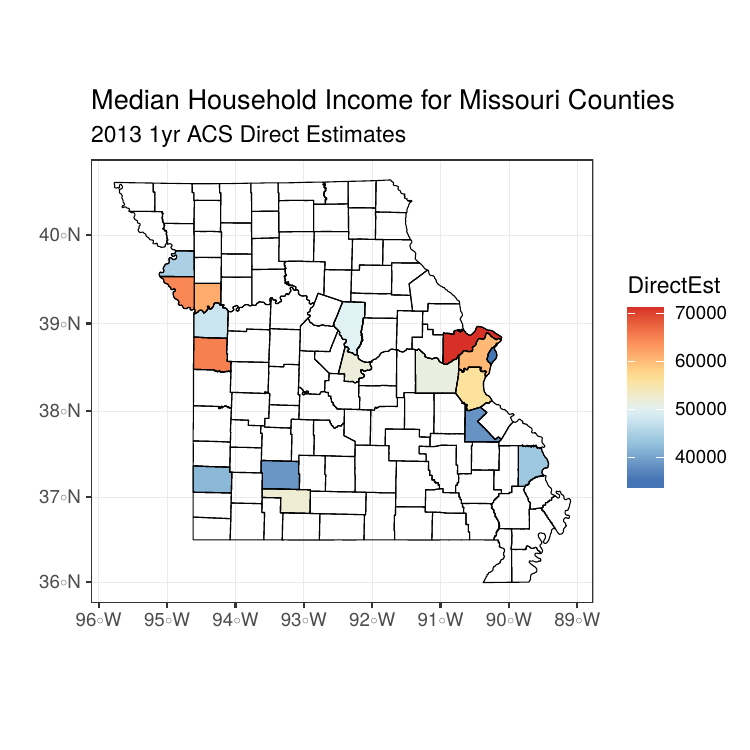}
\includegraphics[width=0.48\textwidth]{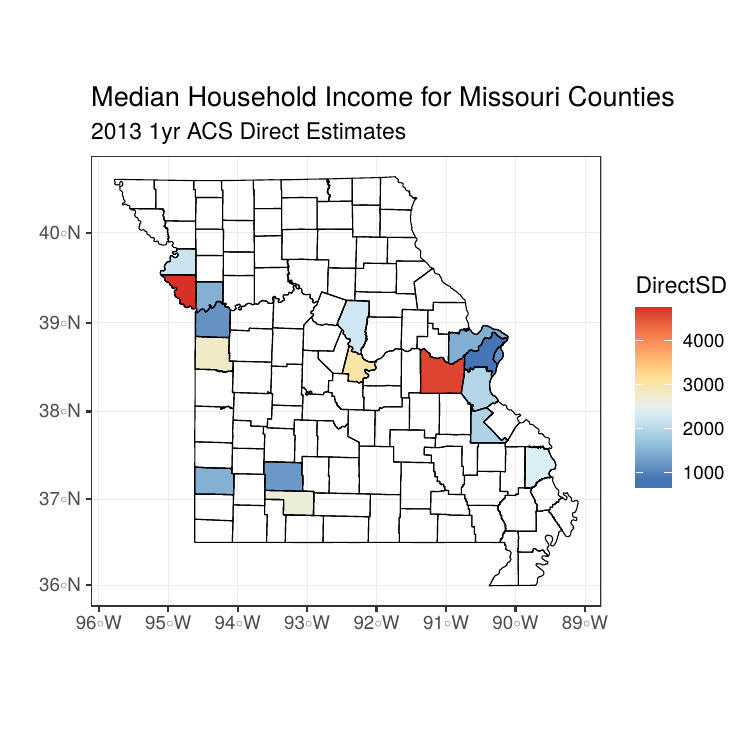} \vspace{-3em} \\
\includegraphics[width=0.48\textwidth]{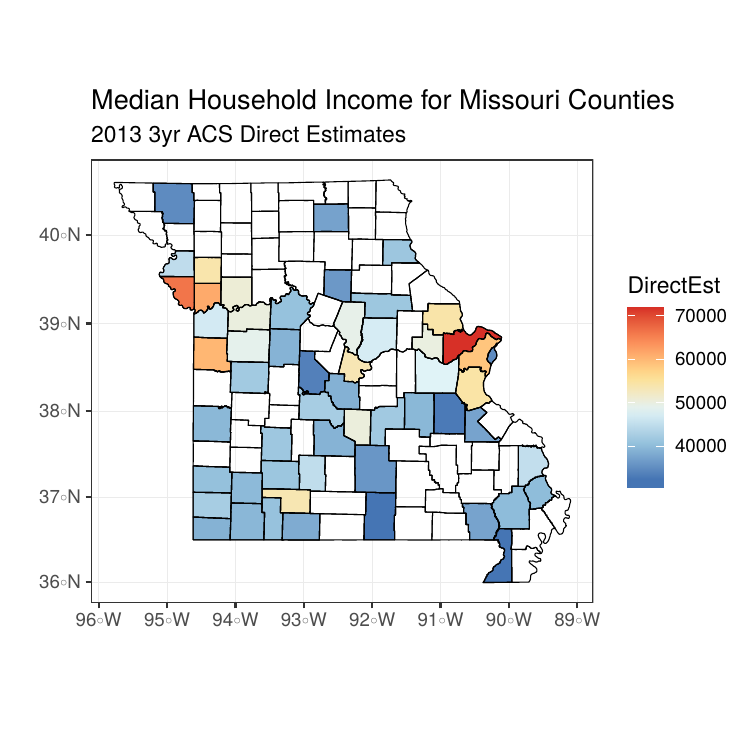}
\includegraphics[width=0.48\textwidth]{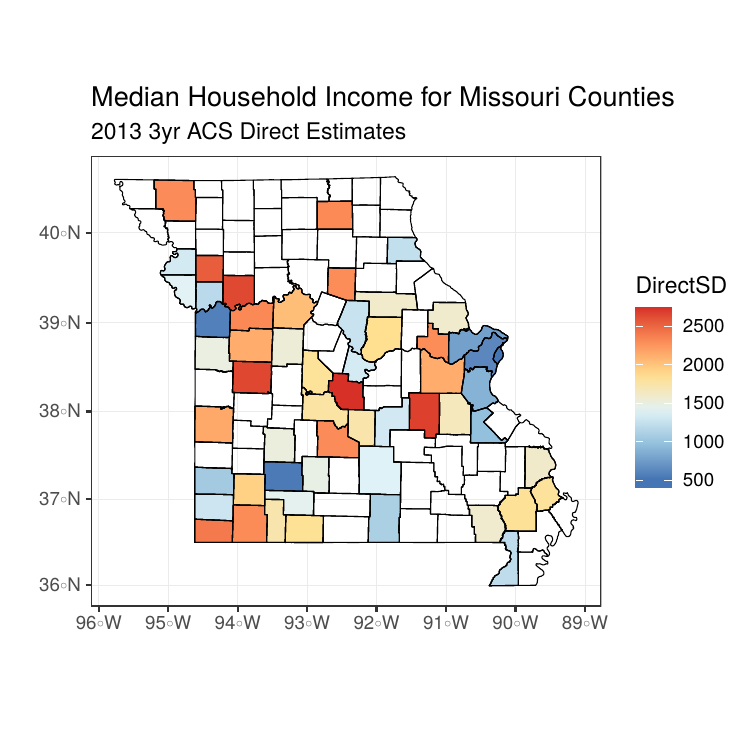} \vspace{-3em} \\
\includegraphics[width=0.48\textwidth]{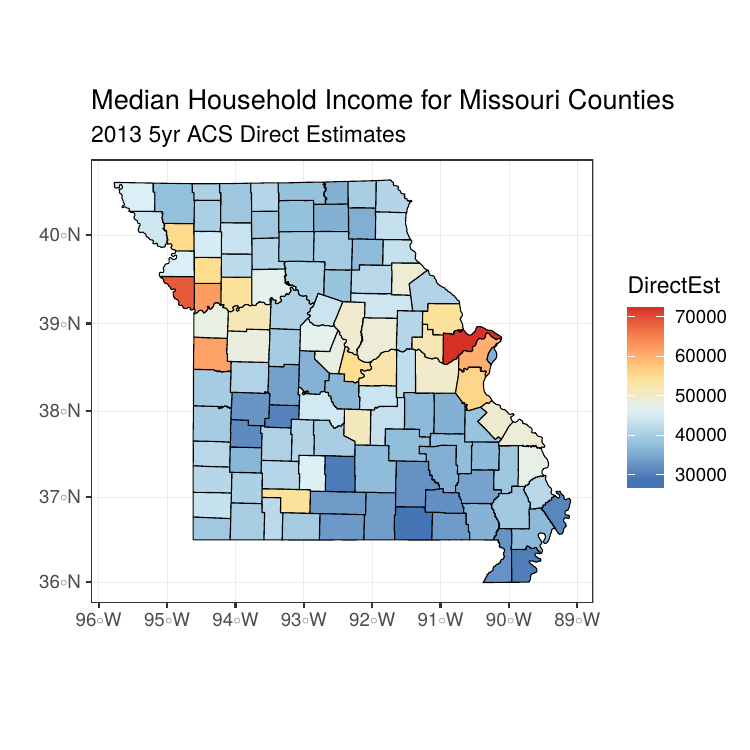}
\includegraphics[width=0.48\textwidth]{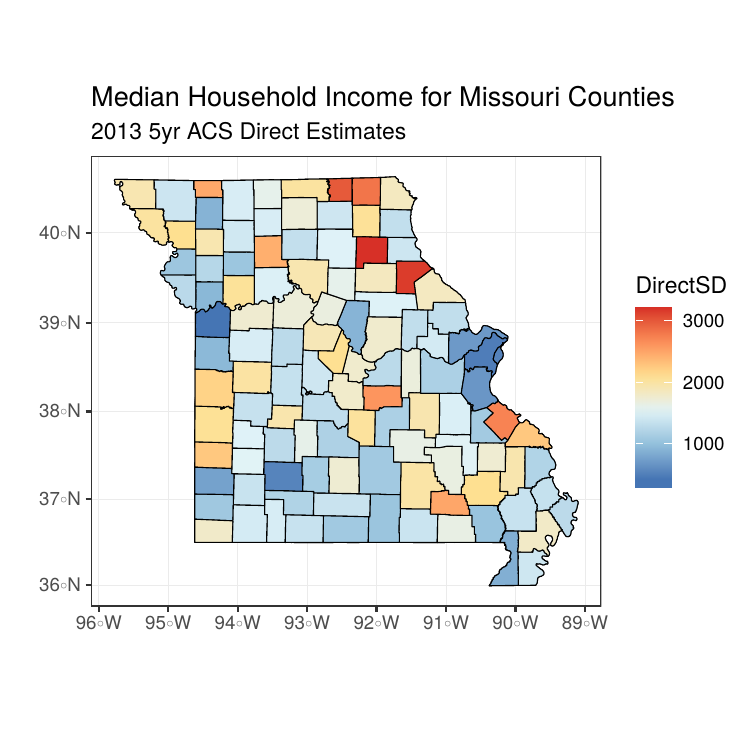}
\vspace{-3em}
\caption{County-level ACS data for median household income in Missouri for the year 2013. The left column shows direct estimates and the right column displays standard errors. The first, second, and third rows correspond to 1-year, 3-year, and 5-year period estimates, respectively. Public ACS estimates were not available for areas with white shading.}
\label{fig:acs-maps}
\end{figure}

\begin{figure}
\centering
\begin{subfigure}{0.35\textwidth}
\centering
\includegraphics[width=\textwidth]{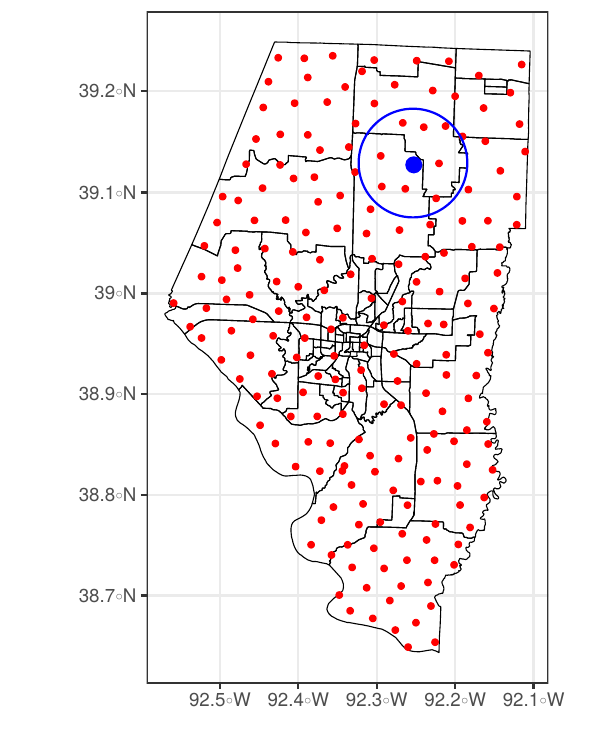}
\caption{Space-filling design.}
\label{fig:spatial-knots-cover}
\end{subfigure}
\begin{subfigure}{0.35\textwidth}
\centering
\includegraphics[width=\textwidth]{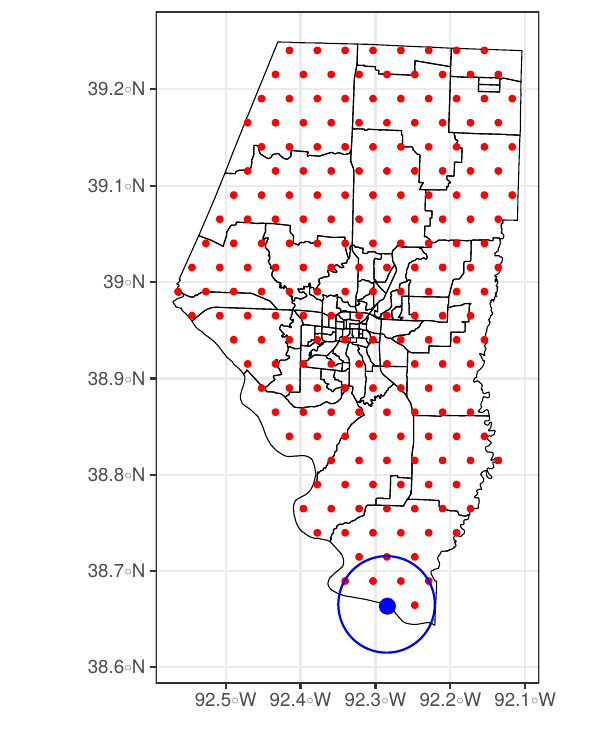}
\caption{Hexagonal sampling.}
\label{fig:spatial-knots-hex}
\end{subfigure}
\caption{Spatial knot points selected for spatio-temporal basis function. Red dots indicate knot points and blue circles display the spatial radius at one particular knot point. Figure~(\subref{fig:spatial-knots-cover}) shows the space-filling design whose radius was $w_s = \text{7,712.70}$, using the quantile calculation and taking $\tilde{w}_s = 1$. Figure~(\subref{fig:spatial-knots-hex}) shows hexagonal sampling, whose radius was $w_s = \text{7,169.13}$ using the same quantile calculation and choice of $\tilde{w}_s$. Note that units of $w_s$ are meters in the selected coordinate system.}
\label{fig:spatial-knots}
\end{figure}

\begin{figure}
\centering
\includegraphics[width=0.5\textwidth]{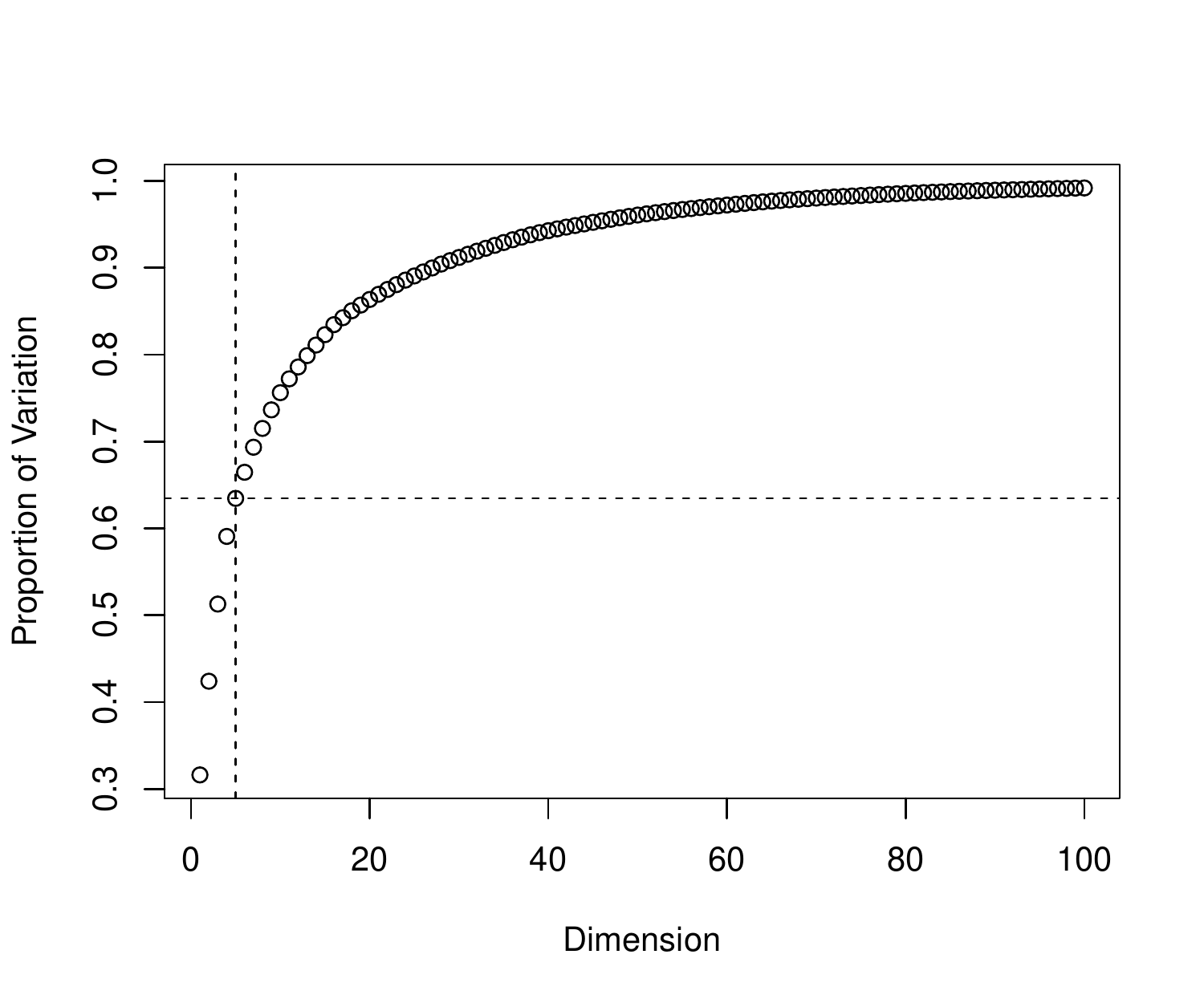}
\caption{Plot of the eigenvalues of $\vec{S}^\top \vec{S}$. The vertical line shows that 5 eigenvectors are needed to capture 65\% of the variation. The y-axis has been truncated to maintain visibility for small dimensions; the total number of eigenvalues is 7,500.}
\label{fig:pca-reduction}
\end{figure}

\begin{figure}
\centering
\includegraphics[width=0.32\textwidth]{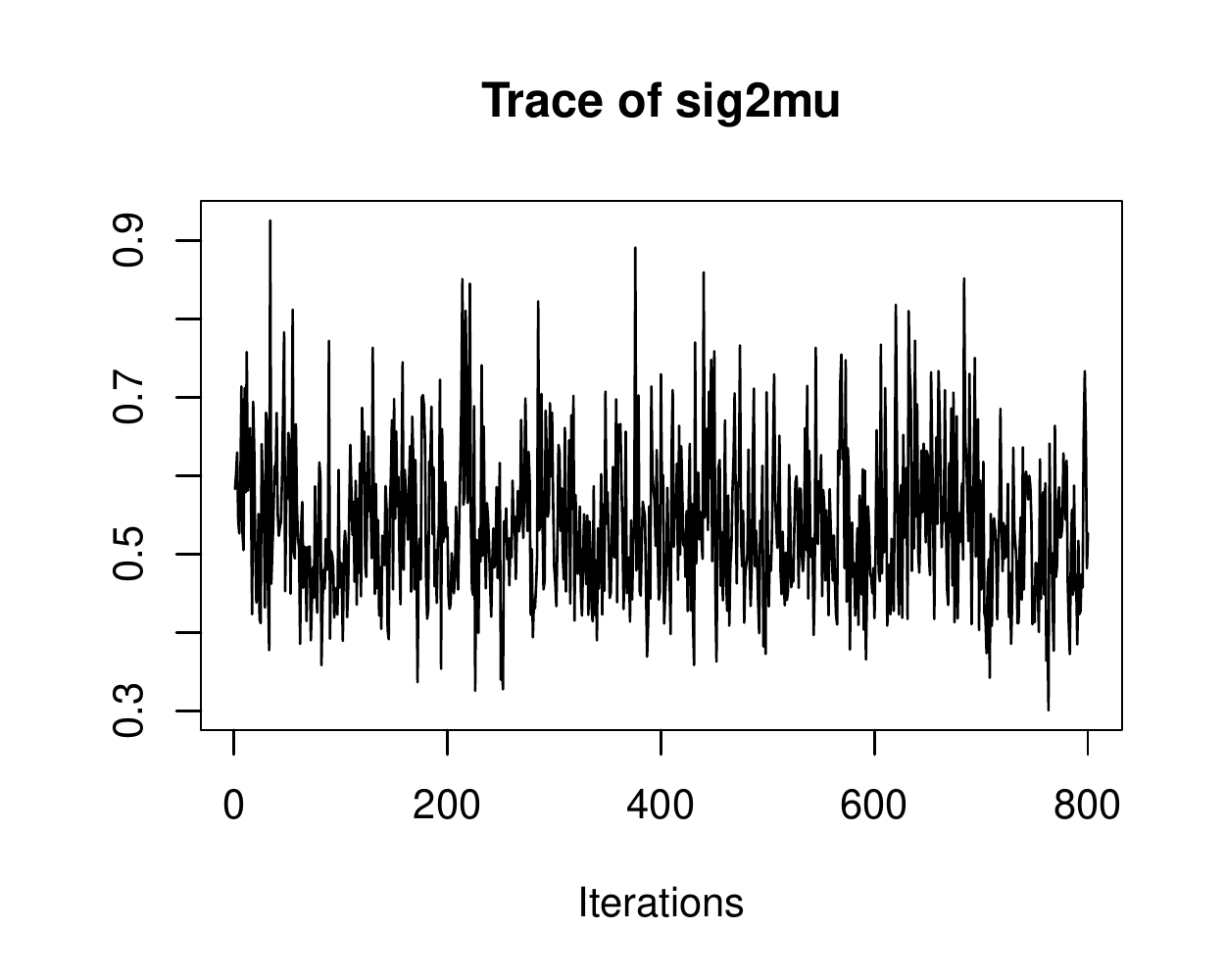}
\includegraphics[width=0.32\textwidth]{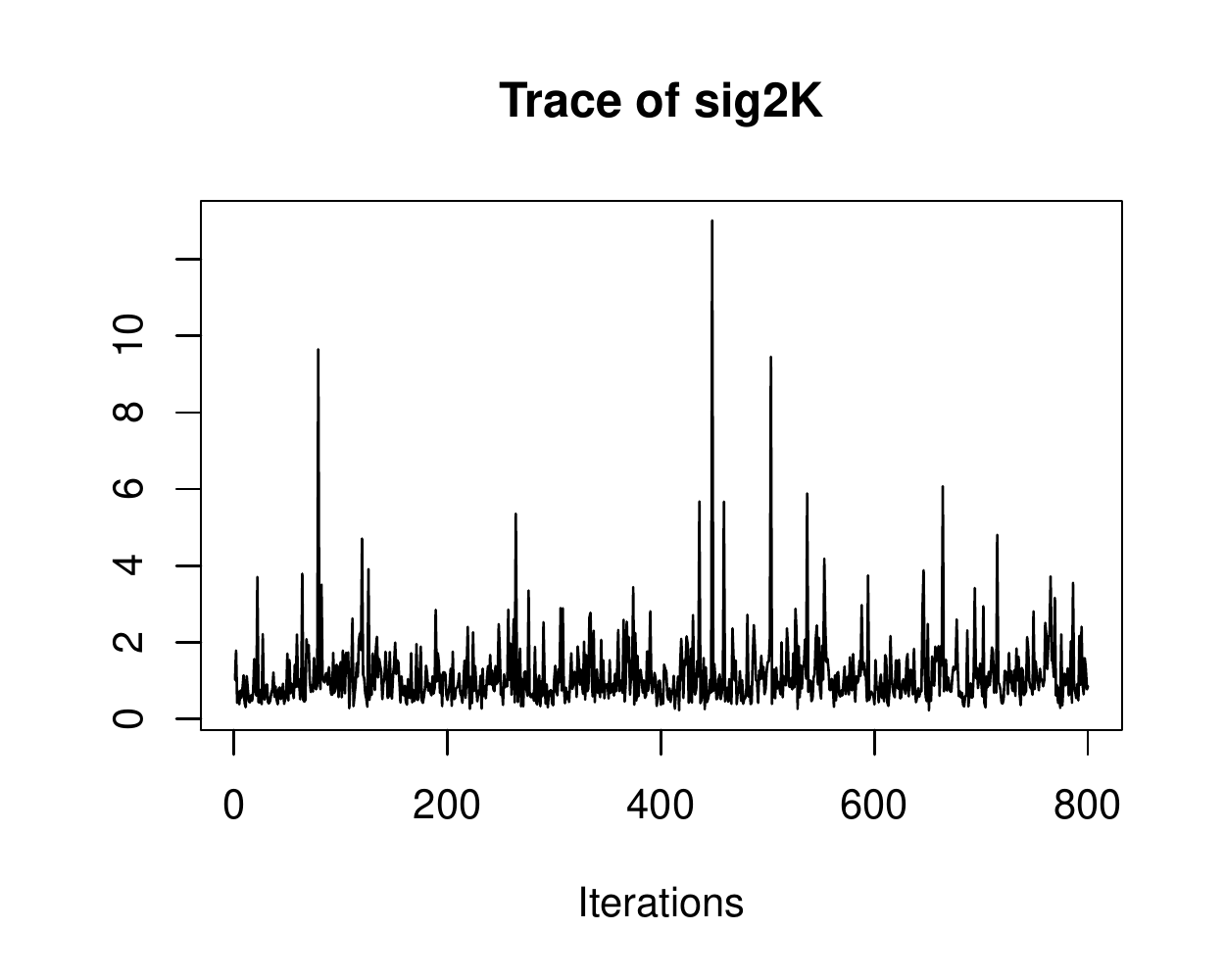}
\includegraphics[width=0.32\textwidth]{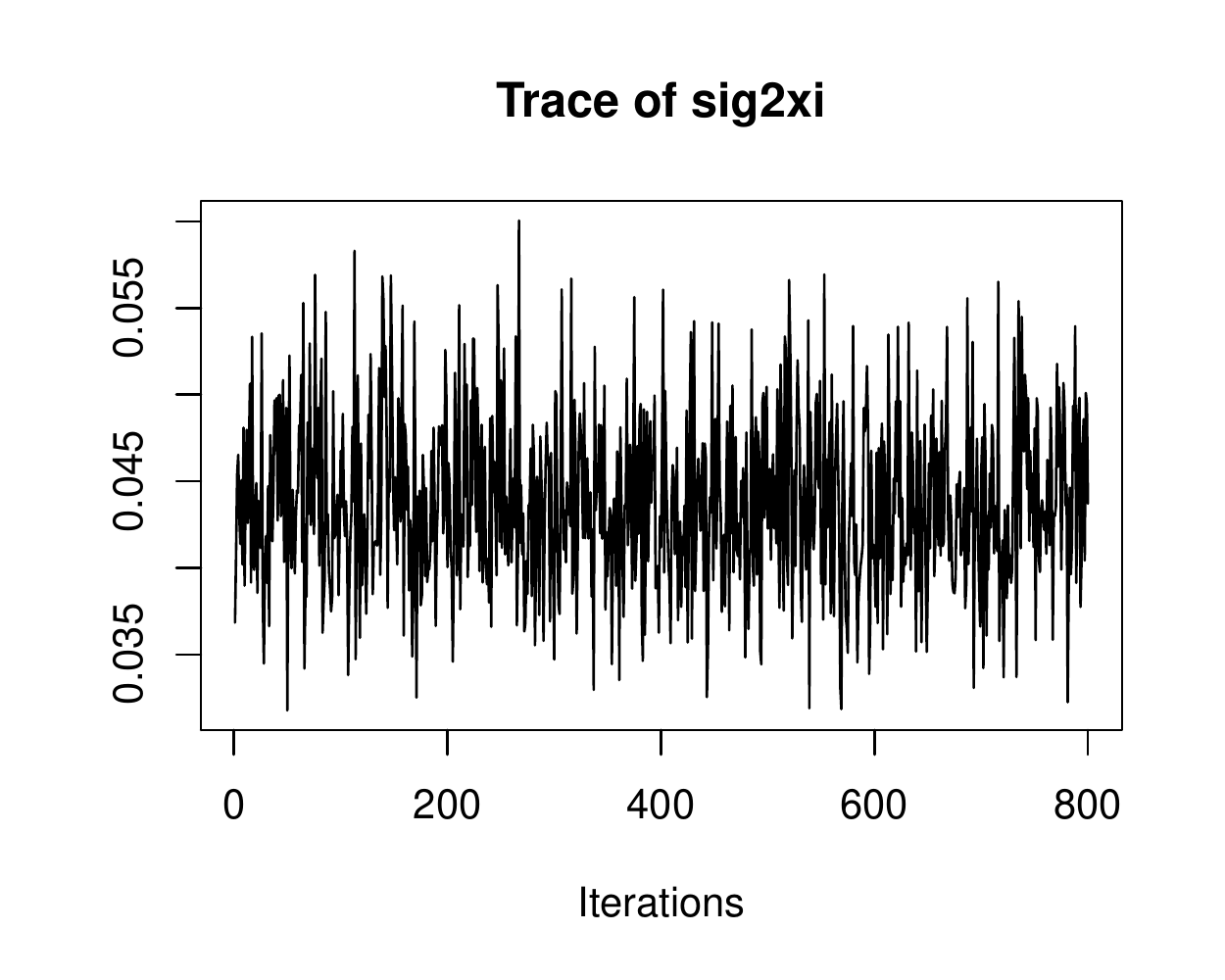}
\includegraphics[width=0.32\textwidth]{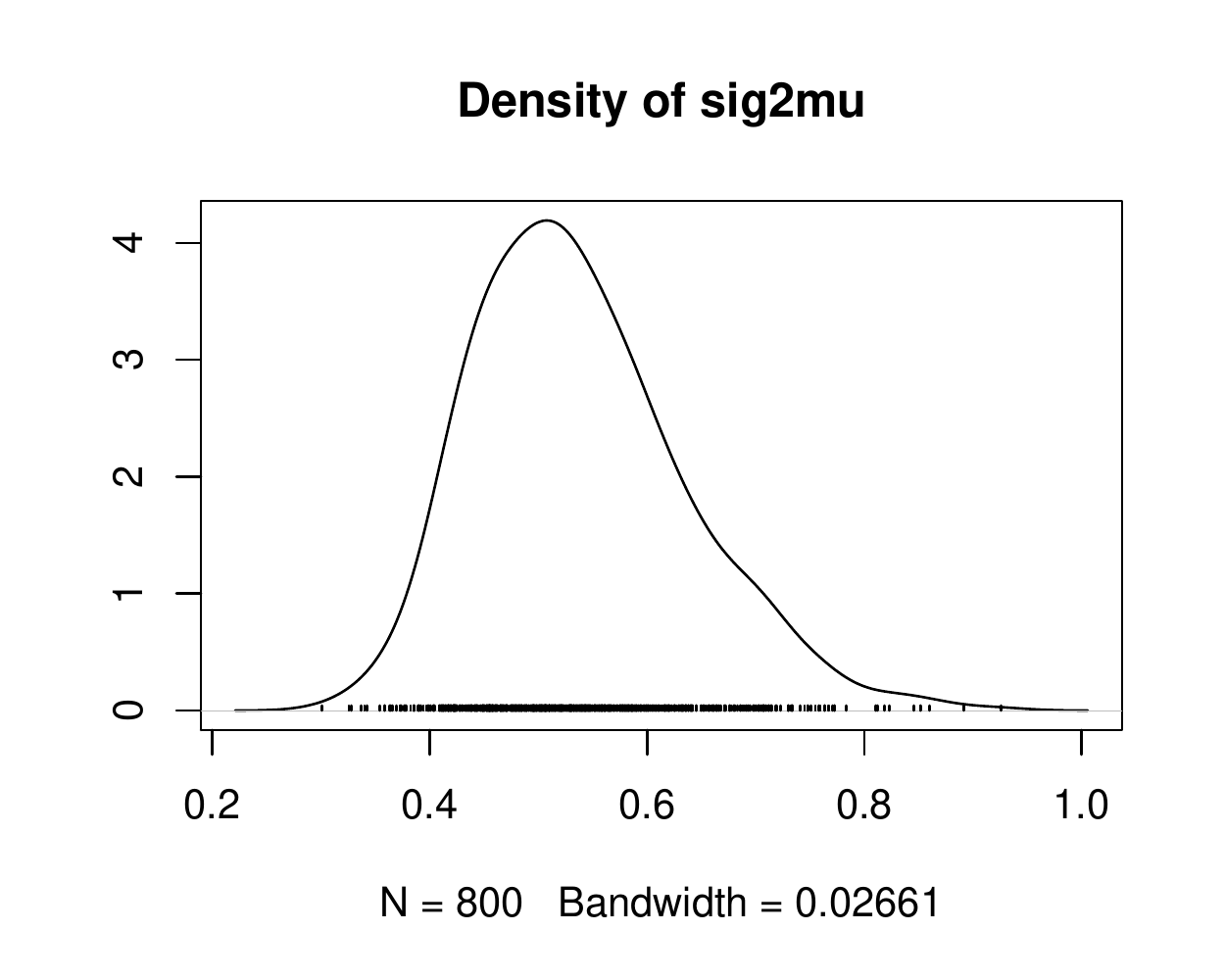}
\includegraphics[width=0.32\textwidth]{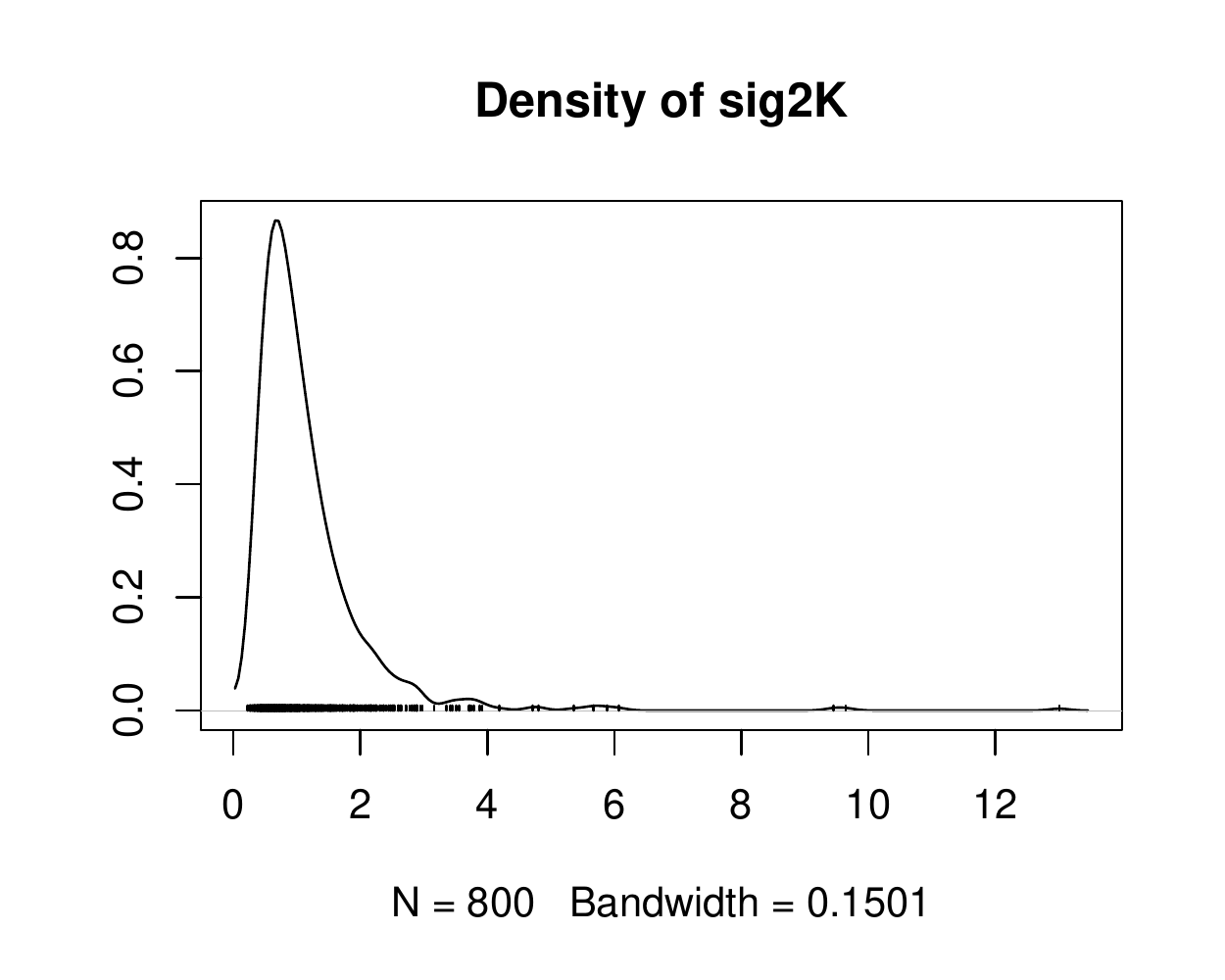}
\includegraphics[width=0.32\textwidth]{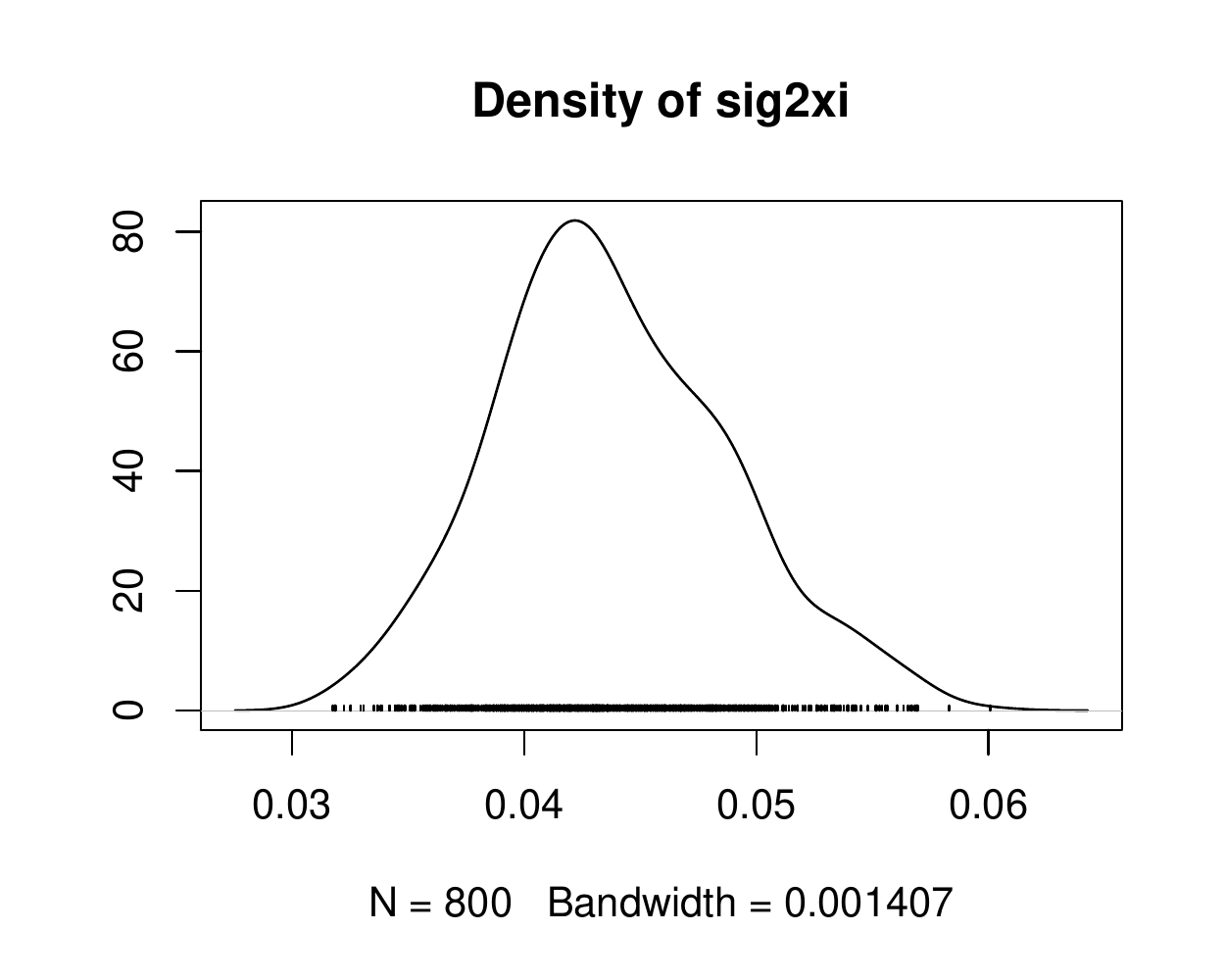}
\caption{Trace and density plots for draws of the variance components $\sigma_\mu^2$, $\sigma_K^2$, and $\sigma_\xi^2$ from the Gibbs sampler.}
\label{fig:trace-varcomps}
\end{figure}

\begin{figure}
\centering
\begin{subfigure}{0.48\textwidth}
\includegraphics[width = 1.10\textwidth]{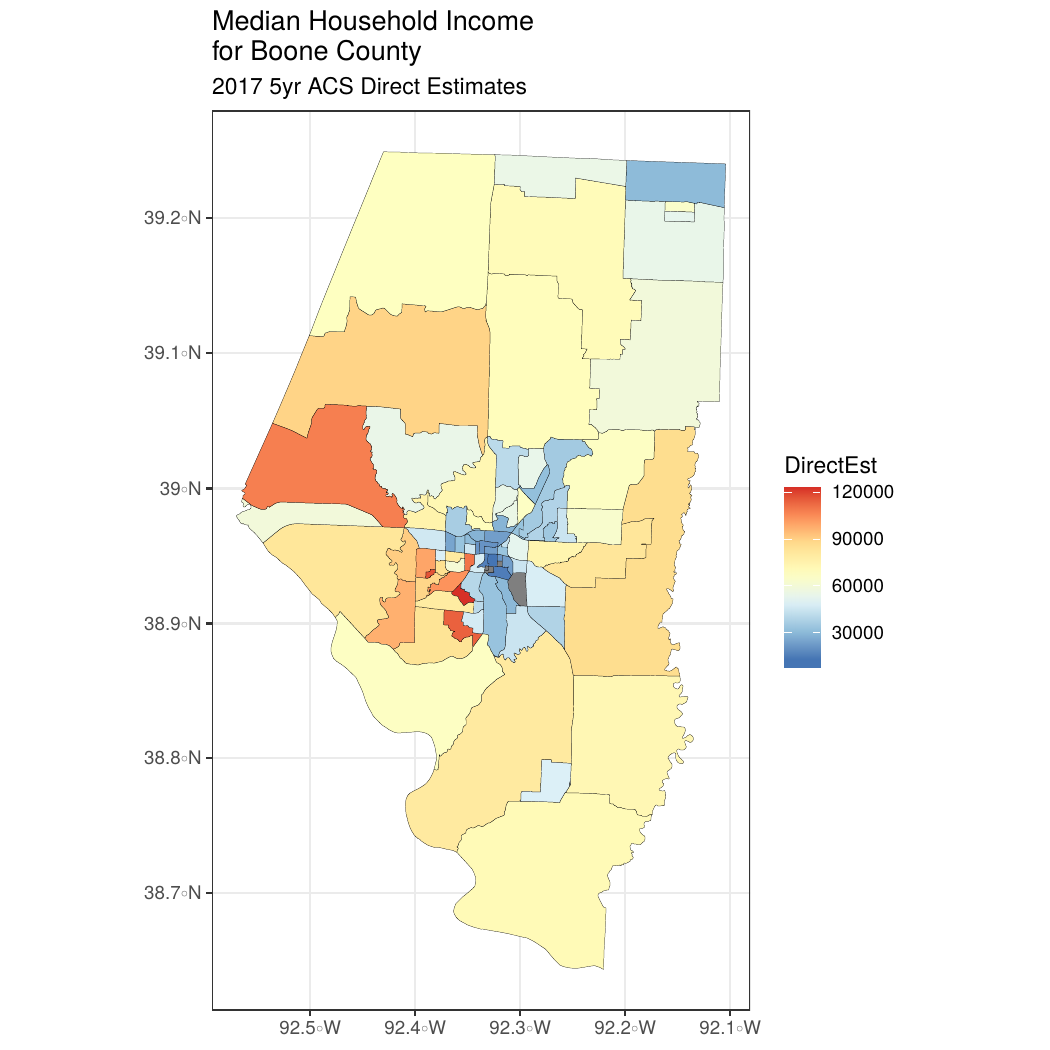}
\caption{}
\label{fig:compare2017-direct}
\end{subfigure}
\begin{subfigure}{0.48\textwidth}
\includegraphics[width = 1.10\textwidth]{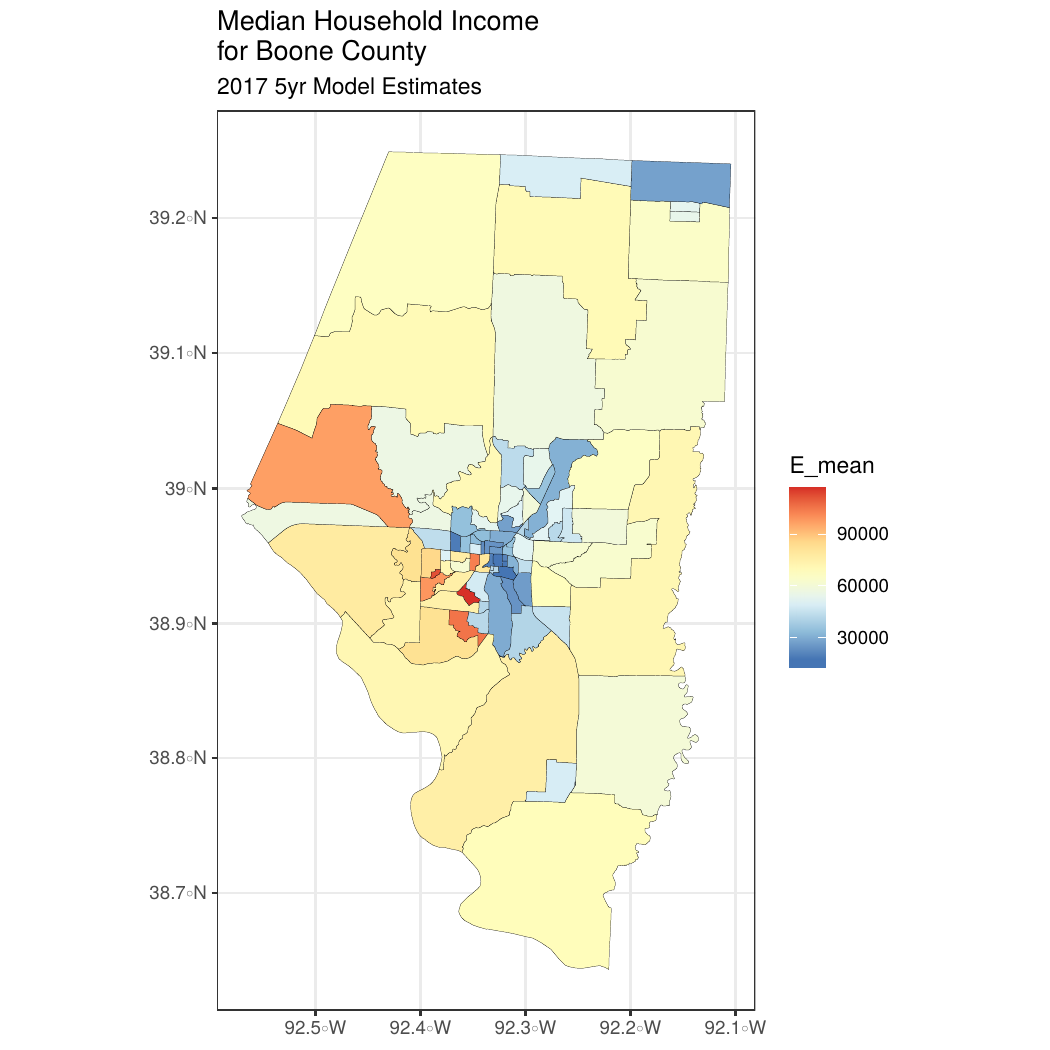}
\caption{}
\label{fig:compare2017-model}
\end{subfigure}
\begin{subfigure}{0.48\textwidth}
\centering
\includegraphics[width = 0.8\textwidth]{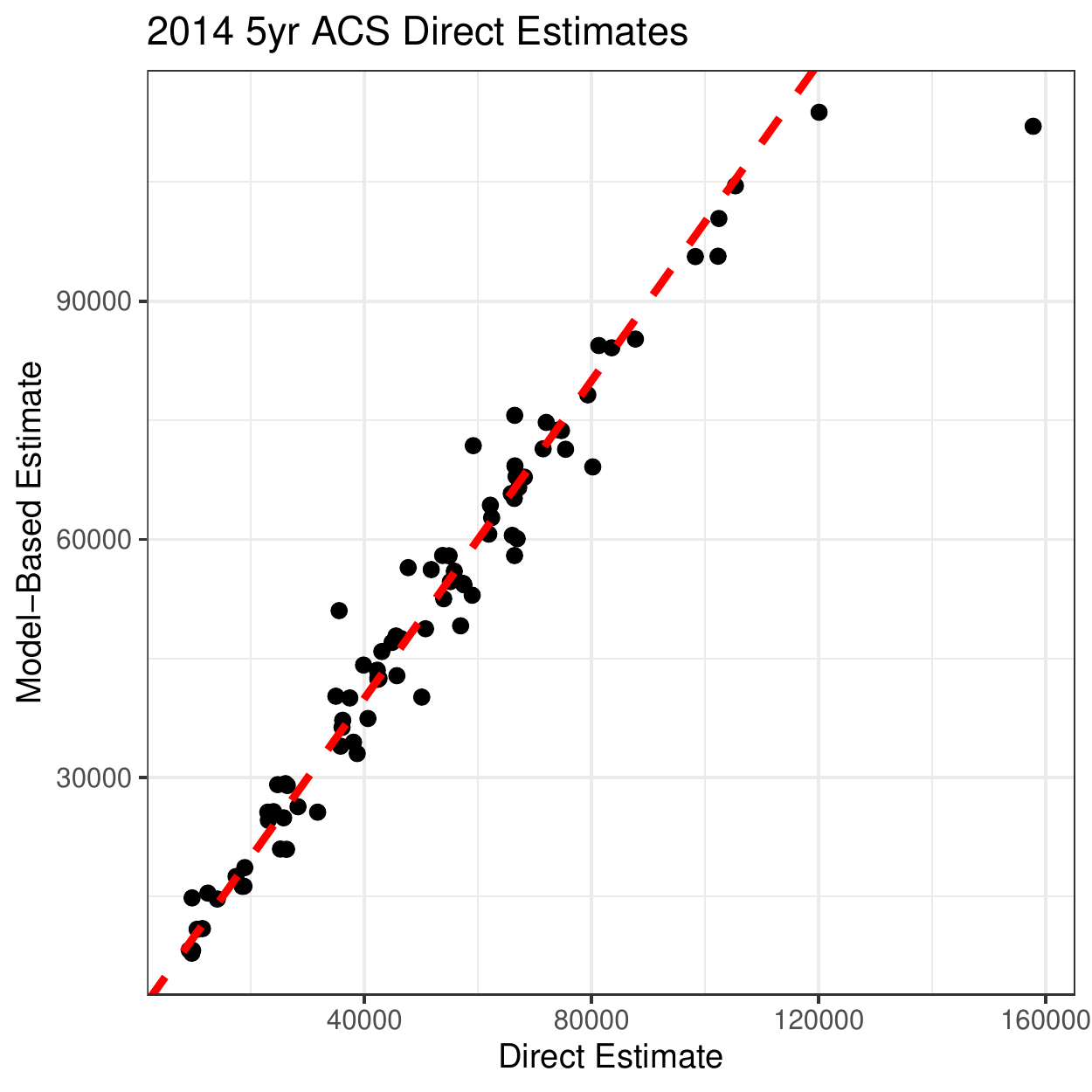}
\caption{}
\label{fig:compare2014-scatter}
\end{subfigure}
\begin{subfigure}{0.48\textwidth}
\centering
\includegraphics[width = 0.8\textwidth]{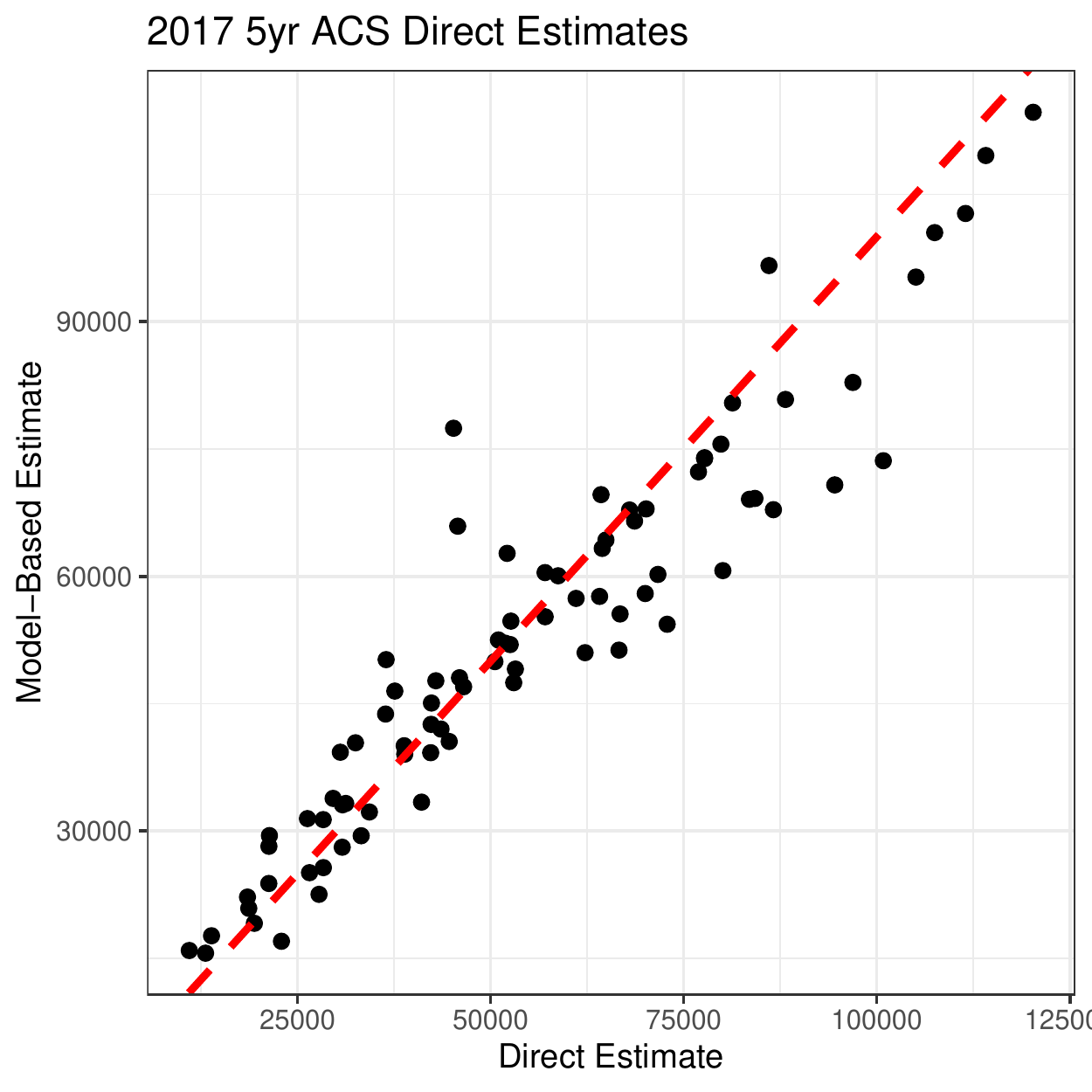}
\caption{}
\label{fig:compare2017-scatter}
\end{subfigure}
\caption{Comparison of direct and model-based ACS 5-year estimates. Figures~\subref{fig:compare2017-direct} and \subref{fig:compare2017-model} show maps based on the two estimates for year 2017. Figures~\subref{fig:compare2014-scatter} and \subref{fig:compare2017-scatter} show scatter plots comparing the two sets of estimates for years 2014 and 2017 respectively.}
\label{fig:compare2017}
\end{figure}

\begin{figure}
\centering
\begin{subtable}{0.99\textwidth}
\centering
\begin{tabular}{l|rrrrrr}
\multicolumn{1}{c|}{Region} &
\multicolumn{1}{c}{Mean} &
\multicolumn{1}{c}{SD} &
\multicolumn{1}{c}{CI Lo} &
\multicolumn{1}{c}{CI Hi} &
\multicolumn{1}{c}{Median} &
\multicolumn{1}{c}{MOE} \\
\hline
Central & 26,931.74 & 1,921.27 & 23,693.43 & 29,925.61 & 26,963.39 & 3,160.21 \\
   East & 44,199.97 & 2,449.92 & 40,217.61 & 48,043.27 & 44,256.71 & 4,029.76 \\
  North & 44,329.41 & 2,861.68 & 39,679.15 & 49,037.54 & 44,202.85 & 4,707.04 \\
  Paris & 20,822.12 & 3,636.90 & 14,965.75 & 26,665.93 & 20,772.98 & 5,982.17 \\
\hline
\end{tabular}
\caption{Estimates based on STCOS model.}
\label{fig:areas-of-interest-table}
\end{subtable}

\vspace{2em}

\begin{subfigure}{1.0\textwidth}
\centering
\includegraphics[width = 1.0\textwidth]{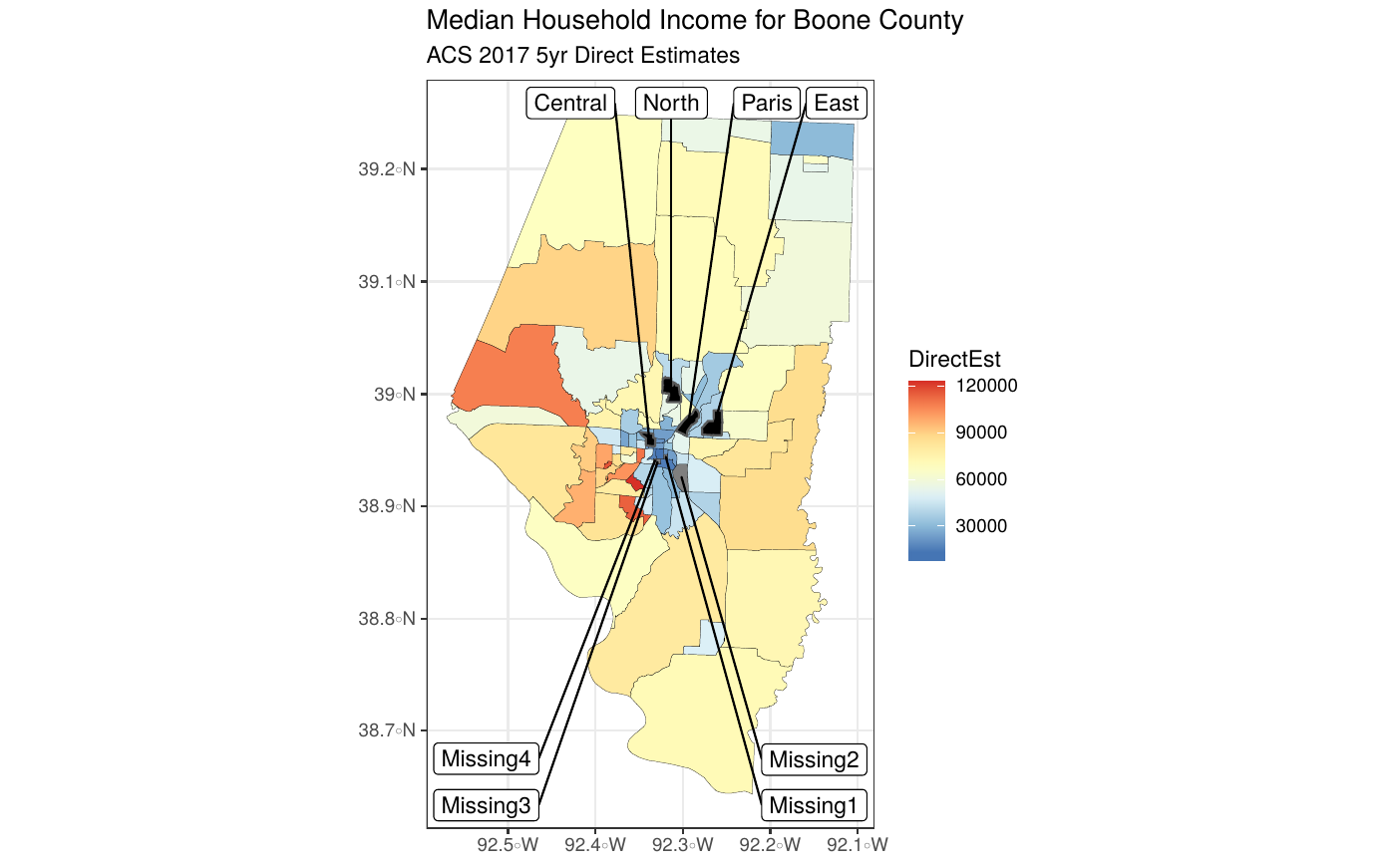}
\caption{Map of 2017 5-year direct estimates.}
\label{fig:areas-of-interest-map}
\end{subfigure}

\caption{
Model-based ACS 5-year estimates for the Central, East, North, and Paris neighborhoods in year 2017 are shown in Table~\ref{fig:areas-of-interest-table}. Figure~\ref{fig:areas-of-interest-map} shows the locations of the four neighborhoods (shaded in black), and year 2017 direct 5-year estimates in Boone County block groups for comparison. Direct estimates were not available for block groups marked as ``Missing'', which are shaded white.}
\label{fig:areas-of-interest}
\end{figure}

\begin{figure}
\centering
\begin{subfigure}{0.48\textwidth}
\centering
\includegraphics[width=\textwidth]{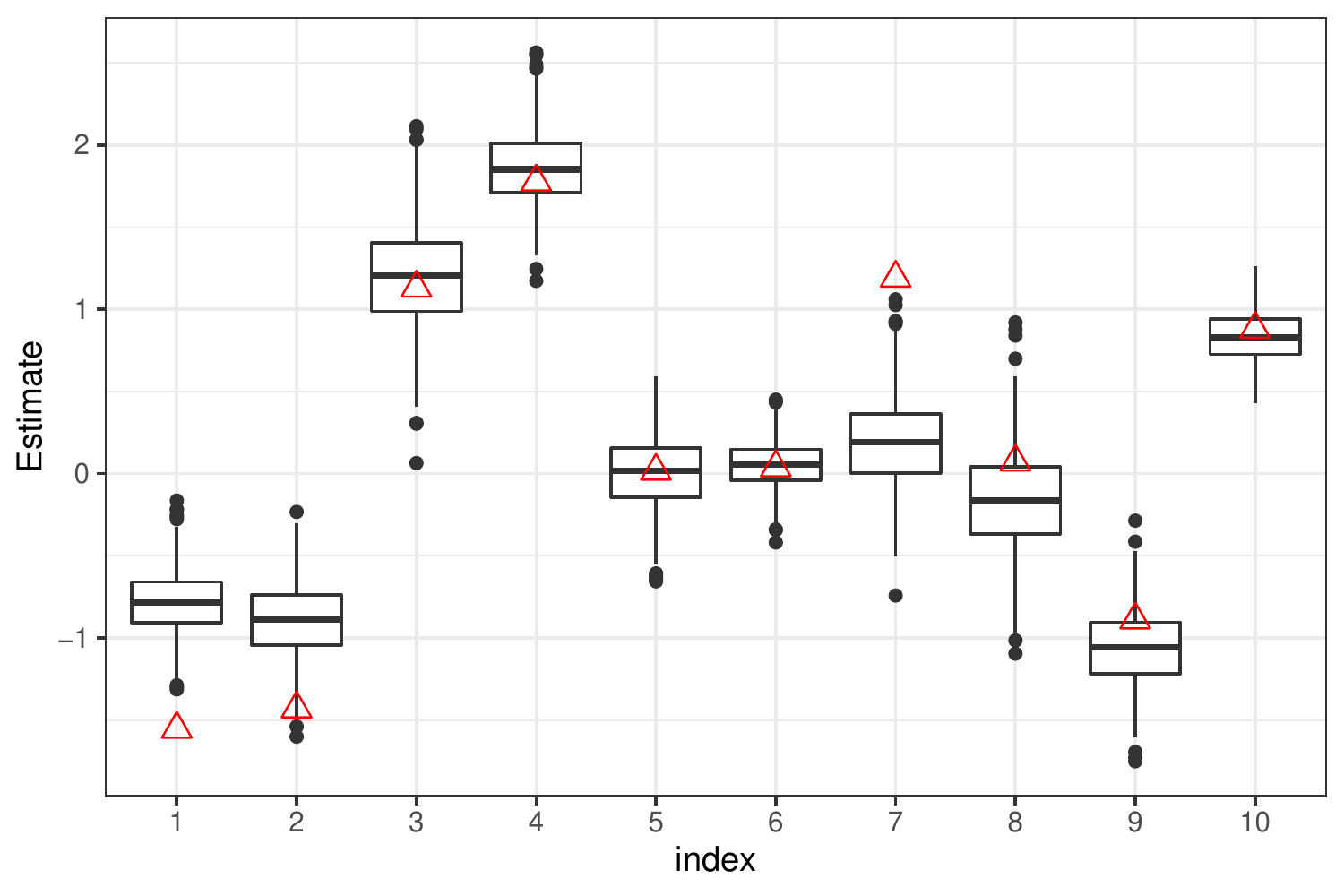}
\caption{}
\label{fig:mle_vs_gibbs_mu_part1}
\end{subfigure}
\begin{subfigure}{0.48\textwidth}
\centering
\centering
\includegraphics[width=\textwidth]{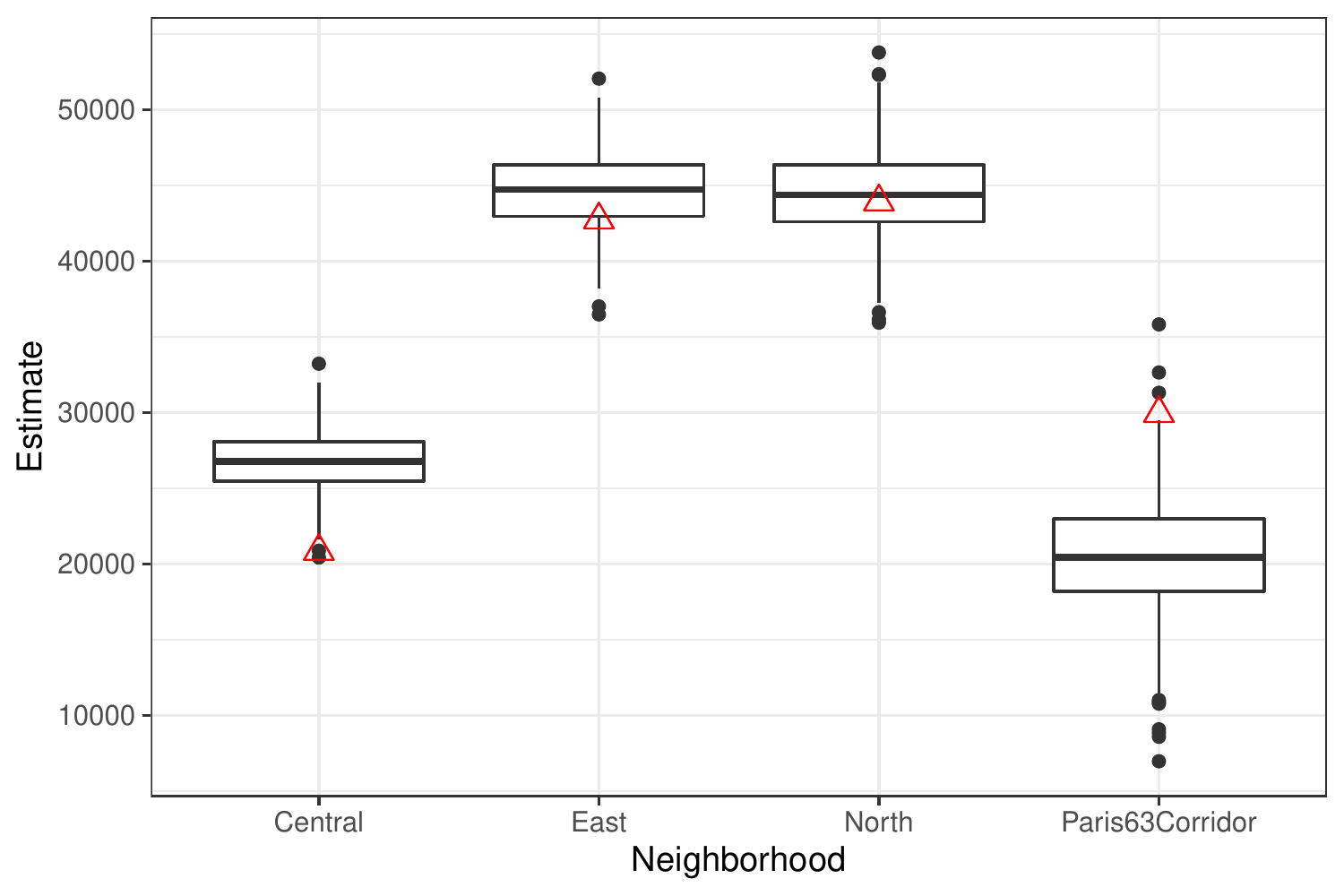}
\caption{}
\label{fig:mle_vs_gibbs_neighb}
\end{subfigure}
\caption{Comparison of MLEs and draws from the Gibbs sampler. The first 10 components of $\vec{\mu}_B$ are displayed in (\subref{fig:mle_vs_gibbs_mu_part1}), while (\subref{fig:mle_vs_gibbs_neighb}) displays estimates for the four neighborhoods which have been transformed to the original scale of the direct estimates. Boxplots correspond to Gibbs sampler draws and red triangles represent MLEs.}
\label{fig:mle_vs_gibbs}
\end{figure}

\clearpage

\appendix

\section{Computational details and proofs}
\label{sec:details}

We will make use of the following well-known property in several places.

\begin{property}
\label{result:vectorization}
If $\vec{A} \in \mathbb{R}^{m \times k}$, $\vec{B} \in \mathbb{R}^{k \times l}$, $\vec{C} \in \mathbb{R}^{l \times n}$, then $\Vec(\vec{A} \vec{B} \vec{C}) = (\vec{C}^\top \otimes \vec{A}) \Vec(\vec{B})$. 
\end{property}

The following proposition gives the explicit solution to the minimization problem stated in \eqref{eqn:approximant-problem}. \citet{BradleyHolanWikle2015AOAS} considers a similar problem featuring a more general objective function but assuming that the columns of $\vec{S}$ are orthonormal. \citet{Higham1988} gives a general discussion of problems involving Frobenius and 2-norm distance minimization.

\begin{proposition}[Frobenius Norm Minimization]
\label{result:frobmin}
Suppose $\vec{S} \in \mathbb{R}^{n \times r}$ has rank $r$ and $\vec{\Sigma} \in \mathbb{R}^{n \times n}$ is positive definite. The minimizer $\vec{X} \in \mathbb{R}^{r \times r}$ of $\lVert \vec{\Sigma} - \vec{S} \vec{X} \vec{S}^\top \rVert_\text{F}$ is $\vec{X} = (\vec{S}^\top \vec{S})^{-1} \vec{S}^\top \vec{\Sigma} \vec{S} (\vec{S}^\top \vec{S})^{-1}$.
\end{proposition}

\begin{proof}
Using Property~\ref{result:vectorization}, we have
\begin{align}
\lVert \vec{\Sigma} - \vec{S} \vec{X} \vec{S}^\top \rVert_{\text{F}}^2
&= \Vec\left[\vec{\Sigma} - \vec{S} \vec{X} \vec{S}^\top \right]^\top
\Vec\left[\vec{\Sigma} - \vec{S} \vec{X} \vec{S}^\top \right] \nonumber \\
&= \left[\Vec(\vec{\Sigma}) - \Vec(\vec{S} \vec{X} \vec{S}^\top) \right]^\top
\left[\Vec(\vec{\Sigma}) - \Vec(\vec{S} \vec{X} \vec{S}^\top) \right] \nonumber \\
&= \left[\Vec(\vec{\Sigma}) - (\vec{S} \otimes \vec{S}) \Vec(\vec{X}) \right]^\top
\left[\Vec(\vec{\Sigma}) - (\vec{S} \otimes \vec{S}) \Vec(\vec{X}) \right] \nonumber \\
&= \lVert \Vec(\vec{\Sigma}) - (\vec{S} \otimes \vec{S}) \Vec(\vec{X}) \rVert_2^2,
\label{eqn:euclidnorm}
\end{align}
where the norm on the last line is the usual 2-norm on $\mathbb{R}^{n^2}$. We recognize the expression in \eqref{eqn:euclidnorm} as a standard least squares minimization whose solution is
\begin{align*}
\Vec(\vec{X})
&= [(\vec{S} \otimes \vec{S})^\top (\vec{S} \otimes \vec{S})]^{-1} (\vec{S} \otimes \vec{S})^\top \Vec(\vec{\Sigma}) \\
&= [(\vec{S}^\top \otimes \vec{S}^\top) (\vec{S} \otimes \vec{S})]^{-1} (\vec{S}^\top \otimes \vec{S}^\top) \Vec(\vec{\Sigma}) \\
&= [ \vec{S}^\top \vec{S} \otimes \vec{S}^\top \vec{S} ]^{-1} \Vec(\vec{S}^\top \vec{\Sigma} \vec{S}) \\
&= [(\vec{S}^\top \vec{S})^{-1} \otimes (\vec{S}^\top \vec{S})^{-1}] \Vec(\vec{S}^\top \vec{\Sigma} \vec{S}) \\
&= \Vec\left[ (\vec{S}^\top \vec{S})^{-1} \vec{S}^\top \vec{\Sigma} \vec{S} (\vec{S}^\top \vec{S})^{-1} \right].
\end{align*}
Therefore, the minimizer is $\vec{X} = (\vec{S}^\top \vec{S})^{-1} \vec{S}^\top \vec{\Sigma} \vec{S} (\vec{S}^\top \vec{S})^{-1}$, as desired.
\end{proof}

\begin{remark}[MLE Computation]
\label{result:mle-computation}
To compute the MLE for the STCOS model, we first note that the likelihood, excluding the parameter model, is
\begin{align*}
f(\vec{Z} \mid \vec{\mu}_B, \sigma_K^2, \sigma_\xi^2)
&= \int \phi(\vec{Z} \mid \vec{H} \vec{\mu}_B + \vec{S} \vec{\eta}, \sigma_\xi^2 \vec{I} + \vec{V}) \cdot
\phi(\vec{\eta} \mid \vec{0}, \sigma_K^2 \vec{K}) d \vec{\eta} \\
&= \phi(\vec{Z} \mid \vec{H} \vec{\mu}_B, \vec{\Delta}) \\
&= (2 \pi)^{-N/2} |\vec{\Delta}|^{-1/2} \exp\left\{ -\frac{1}{2} (\vec{Z} -
\vec{H} \vec{\mu}_B)^\top \vec{\Delta}^{-1} (\vec{Z} - \vec{H} \vec{\mu}_B) \right\},
\end{align*}
where $\vec{\Delta} = \sigma_\xi^2 \vec{I} + \vec{V} + \sigma_K^2 \vec{S} \vec{K} \vec{S}^\top$. Given $\sigma_K^2$ and $\sigma_\xi^2$, the likelihood is maximized by the weighted least squares estimator $\hat{\vec{\mu}}_B = (\vec{H}^\top \vec{\Delta}^{-1} \vec{H})^{-1} \vec{H}^\top \vec{\Delta}^{-1} \vec{Z}$. To estimate the unknown $\sigma_K^2$ and $\sigma_\xi^2$, we carry out numerical maximization on the partially maximized log-likelihood
\begin{align*}
\ell(\sigma_K^2, \sigma_\xi^2) =
-\frac{N}{2} \log(2 \pi) -\frac{1}{2} \log |\vec{\Delta}| -\frac{1}{2} (\vec{Z} -
\vec{H} \hat{\vec{\mu}}_B)^\top \vec{\Delta}^{-1} (\vec{Z} - \vec{H} \hat{\vec{\mu}}_B).
\end{align*}
To enforce the constraints that $\sigma_K^2 > 0$ and $\sigma_\xi^2 > 0$, we optimize over $(\vartheta_1, \vartheta_2) \in \mathbb{R}^2$ and take $\sigma_K^2 = \exp(\vartheta_1)$, $\sigma_\xi^2 = \exp(\vartheta_2)$.
\end{remark}

\section{Supplementary: Scaling to Large Datasets}
\label{sec:scaling}

The demonstration in Section~\ref{sec:columbia} was carried out with a relatively small dataset based on $N = 421$ total observations in the source supports and $n_B = 85$ areas in the fine-level support. Here we will list some challenges that may be encountered when scaling to larger datasets.

\begin{myenumerate}
\item Work to compute the basis function \eqref{eqn:basis-spt} is proportional to the number of areas in the supports, the length of their lookback periods, and the number of Monte Carlo repetitions requested. Basis computations are independent across supports, or across areas within the same support, and therefore can be computed in parallel.

\item Time to compute the overlap matrix $\vec{H}$ also increases proportionally with the number of areas in supports. This is less substantial than basis function computation, but can also be parallelized across supports.

\item The Monte Carlo approximation for \eqref{eqn:basis-spt} utilizes a rejection sampling method to draw a sample from each area $A$, first drawing a point from the bounding box that surrounds $A$, then accepting the point only if it belongs to $A$ itself. This method has difficulty accepting samples when $|A|$ is very small relative to the bounding box; for example, if $A$ is a thin rectangle.

\item We used PCA to reduce the dimension of $\vec{S}$ based on a spectral decomposition of $\vec{S}^\top \vec{S}$. Before dimension reduction, the matrix $\vec{S}^\top \vec{S}$ is typically sparse, but the dimension can be quite large depending on the number of knot points used. Here it can be helpful to request a limited number of eigenvalue/eigenvector pairs; for example, this can be done using sparse matrices with the \code{RSpectra} package \citep{RSpectra2019}. After dimension reduction, neither $\vec{S}$ nor $\vec{S}^\top \vec{S}$ are typically sparse.

\item Although the CAR precision matrix $\vec{Q}$ is typically sparse, its inverse is dense when it exists. However, $\vec{Q}^{-1}$ is currently used only in the construction of $\vec{K}$.

\item The Gibbs sampler in Algorithm~\ref{alg:gibbs-sampler} involves repeated operations with $n_B \times n_B$ and $r \times r$ matrices: $r$ can be kept to a manageable size using the suggested dimension reduction, but $n_B$ depends on the choice of fine-level support.

\item For both the Gibbs and Stan sampling, storing all of $\vec{\mu}_B \in \mathbb{R}^{n_B}$, $\vec{\xi} \in \mathbb{R}^{N}$, and $\vec{\eta} \in \mathbb{R}^r$ for every iteration of the sampler can become a memory/storage burden. The vector $\vec{\xi}$ has not been used in post-processing and need not be stored. However, $\vec{\mu}_B$ and $\vec{\eta}$ are both utilized in post-processing. For very large $n_B$, it may be more efficient to compute desired functions of the draws within the sampler rather than saving the draws for later use.

\item To support a sparse representation for a large $\vec{H}$ matrix in Stan, we can make use of Stan's \code{csr_matrix_times_vector} function to compute $\vec{H} \vec{\mu}_B$.

\item The MLE computation in Remark~\ref{result:mle-computation} utilizes the $N \times N$ covariance matrix 
\(
\vec{\Delta} = \sigma_K^2 \vec{S} \vec{K} \vec{S}^\top + \sigma_{\xi}^2 \vec{I} + \vec{V}
\)
of the marginal distribution of $\vec{Z}$. This matrix is also utilized in a Stan sampler with the quantities $\vec{\xi}$ and $\vec{\eta}$ integrated out. We may not be able to explicitly construct $\vec{\Delta}$ if $N$ is very large. As previously noted, the matrix $\vec{S}$ will be dense if its dimension has been reduced via PCA, so that $\vec{\Delta}$ in turn will also be dense. Some additional matrix algebra may assist in computing the likelihood. For example, the Sherman-Morrison-Woodbury identity yields
\begin{align}
\vec{\Delta}^{-1} = \vec{U}^{-1} - \vec{U}^{-1}
\vec{S}
\left[ \sigma_K^{-2} \vec{K}^{-1} + \vec{S}^\top \vec{U}^{-1} \vec{S} \right]^{-1}
\vec{S}^\top
\vec{U}^{-1}
\label{eqn:smw}
\end{align}
where $\vec{U} = \sigma_{\xi}^2 \vec{I} + \vec{V}$. Using \eqref{eqn:smw}, the quadratic form $\vec{Z}^\top \vec{\Delta}^{-1} \vec{Z}$ can be computed without forming (or inverting) any $N \times N$ matrices.
\end{myenumerate}

To illustrate the methods on a larger dataset, code for a larger scale data analysis is also provided in the supplemental materials for this article. Here the analysis is structured similarly to \citet{JSM2017-STCOS}. There are 17 source supports with counties in the continental United States and median household income as the ACS variable of interest. We use 1-year estimates for years 2011--2017, 3-year estimates for years 2011--2013, and 5 year estimates for years 2011--2017. The 2017 county geography was taken to be the fine-level support, and the target support was taken to be 2017 congressional districts. This yields $N = 32943$ observations in the source supports and $n_B = 3105$ fine-level support areas. For the dimension reduction of $\vec{S}$, we took $r = 56$. The code was run on an Intel Core i7-2600 3.40GHz with 4 cores and 8 GB of memory using \pkg{stcos} version 0.3.0. The following highlights were noted.

\begin{myenumerate}
\item Overall run time was about 5.5 hours. This consisted of: 1 minute to download and assemble data using Census API and \pkg{tigris} package as in Section~\ref{sec:assemble}, 167 minutes to prepare the analysis, similarly to Section~\ref{sec:prepare}, 75 minutes to run the Gibbs sampler and produce plots of estimates for congressional districts, and 88 minutes to run the Stan sampler and produce equivalent plots.
\item Overlap matrix $\vec{H}$ took 9 minutes to construct.
\item Design matrix $\vec{S}$ took 83 minutes to construct from basis functions. Each 1-year, 3-year, and 5-year source support took about 1 minute, 4 minutes, and 9 minutes, respectively.
\item Design matrix $\vec{S}^*$  took about 70 minutes to construct.
\item Constructing $\vec{K}$ via \eqref{eqn:random-walk} and \eqref{eqn:spatial-only} took 2 minutes and 7 seconds, respectively.
\item Gibbs sampler with 10000 iterations, burn-in of 2000, and thinning to save 1 of every remaining 10 draws took 71 minutes.
\item Stan sampling with 2000 iterations, burn-in of 1200, and no thinning took 88 minutes.
\end{myenumerate}
Note that these results are intended to give a rough idea of run times, and may vary depending on hardware, installed libraries, compilers, and many other factors. Improvements may be possible in future versions of \pkg{stcos}, or in the analysis code itself, to improve scalability.

\end{document}